\documentclass[english]{article}
\usepackage{amsmath,amssymb}
\usepackage[english]{babel}
\usepackage{fullpage}
\usepackage{graphics,graphicx}

\usepackage[ruled,vlined]{algorithm2e}
\newcommand{\R}{\ensuremath{{\rm \mathbb R}}}
\newcommand{\N}{\ensuremath{{\rm \mathbb N}}}

\newcommand{\dest}{\ensuremath{{d}}}
\DeclareMathOperator*{\argmin}{argmin}
\DeclareMathOperator*{\argmax}{argmax}

\newcommand{\SolveSPTG}{\mbox{\bf\ttfamily SolveSPTG}}

\newcommand{\SolvePTG}{\mbox{\bf\ttfamily SolvePTG}}
\newcommand{\NextEventPoint}{\mbox{\bf\ttfamily NextEventPoint}}
\newcommand{\ExtendedDijkstra}{\mbox{\bf\ttfamily ExtendedDijkstra}}

\newcommand{\cost}{\mbox{\rm cost}}
\newcommand{\Plays}{\mbox{\rm Plays}}
\newcommand{\countt}{\mbox{\rm count}}

\newcommand{\StrategyIteration}{\mbox{\bf\ttfamily StrategyIteration}}
\newcommand{\EventPoints}{L}

\newtheorem{theorem}{Theorem}[section]
\newtheorem{lemma}[theorem]{Lemma}
\newtheorem{remark}[theorem]{Remark}

\newtheorem{corollary}[theorem]{Corollary}
\newtheorem{definition}[theorem]{Definition}
\newtheorem{conjecture}[theorem]{Conjecture}
\newenvironment{proof}{{\bf \noindent Proof:\ }}{\hfill$\Box$\medskip}

\usepackage{ifpdf}
\ifpdf
\else
\fi

\begin{document}

\title{A Faster Algorithm for Solving One-Clock \\ Priced Timed Games\thanks{
Work 
supported
by the Sino-Danish Center for the Theory of Interactive Computation,
funded by the Danish National Research Foundation and the National
Science Foundation of China (under the grant 61061130540). The authors acknowledge support from the Center for research in
the Foundations of Electronic Markets (CFEM), supported by the Danish
Strategic Research Council. Thomas
Dueholm Hansen is a recipient of the Google Europe Fellowship in
Game Theory, and this research is supported in part by this
Google Fellowship.}}
\author{
Thomas Dueholm Hansen\thanks{Department of Computer Science,
Aarhus University, Denmark. E-mail:
{\tt \{tdh, rij, pbmiltersen\}@cs.au.dk}.}
\and Rasmus Ibsen-Jensen$^\dagger$
\and Peter Bro Miltersen$^\dagger$
}
\date{}

\maketitle

\begin{abstract}
One-clock priced timed games is a class of two-player, zero-sum,
continuous-time games that was defined and thoroughly studied in
previous works. We show that one-clock priced timed games can be
solved in time $m 12^{n}n^{O(1)}$, where $n$ is the number of states and
$m$ is the number of actions. The best previously known time bound for
solving one-clock priced timed games was
$2^{O(n^2+m)}$, due to Rutkowski. 
For our improvement, we introduce and study a new algorithm for solving
one-clock priced timed games, based on the sweep-line technique from
computational geometry and the strategy iteration paradigm from the
algorithmic theory of Markov decision processes.
As a corollary, we also improve the analysis
of previous algorithms due to Bouyer, Cassez, Fleury, and Larsen; and
Alur, Bernadsky, and Madhusudan.


%

%
\end{abstract}

\section{Introduction}

Priced timed {\em automata} and priced timed {\em games} are classes
of one-player and two-player zero-sum real-time games played on finite
graphs that were defined and thoroughly studied in previous works 
\cite{AD94, BFHLPRV00, ATP01, TMM02, ABM04, BCFL04, BBR05,
  BBM06, BLMR06, JT07, Rutkowski2011}. Synthesizing (near-)optimal strategies for
priced timed games has many practical applications in embedded systems
design; we refer to the cited papers for references.

\smallskip\noindent{\bf Informal description of priced timed games.} Informally (for formal definitions, see the sections below), a priced
timed game is played by two players on a finite directed labeled multi-graph. The
vertices of the graph are called {\em states}, with some states
belonging to Player 1 (or the Minimizer) and the other states belonging to
Player 2 (or the Maximizer). We shall denote by $n$ the total number
of states of the game under consideration and $m$ the total number of
arcs (actions). Player 1
is trying to play the game to termination as cheaply as possible,
while Player 2 is trying to make Player 1 pay as dearly as possible
for playing. At any point in time, some particular state is the {\em
  current} state. The player controlling the state decides when to
leave the current state and which arc to follow when doing so. For
each arc, there is an associated {\em cost}. Each state has an
associated {\em rate} of expense per time unit associated with waiting
in the state. The above setup is further refined by the introduction
of a finite number of {\em clocks} that can informally be thought of
as ``stop watches''. In particular, some arcs may have associated a
{\em reset} event for a clock. If the corresponding transition is
taken, that clock is reset to 0.  Also, an arc may have an associated
clock and time interval. When the arm of the clock is in the interval,
the corresponding transition can be taken; otherwise it can not. With
three or more clocks, the problem of solving priced timed games is
known not to be computable \cite{BBM06}. In this paper, we focus on the computable case of solving
one-clock priced timed games. We shall refer to these as PTGs. We
shall furthermore single out an important, particularly clean, special
case of PTGs. We shall refer to this class as {\em simple priced timed
  games}, SPTGs.  In an SPTG, time runs from 0 to 1, the single clock
is never reset, and there are no restrictions on when transitions may
be taken. A slightly more general class of games was called ``[0,1]-PTGs without resets'' by Bouyer et al. \cite{BLMR06}. 

\smallskip\noindent{\bf Values and strategies.} As is the case in general for two-player zero-sum games, informally, a
priced timed game is said to have a {\em value} $v$ if Player 1 and
Player 2 are both able to guarantee (or approximate arbitrarily well)   a
total cost of $v$ when the game is played. The guarantees are obtained
when players commit to (near-)optimal {\em strategies} when playing
the game. Player 1, who is trying to minimize cost, may
(approximately) guarantee the value from above, while Player 2, who is
trying to maximize cost, may (approximately) guarantee the value from
below. Clearly, in general, the value of a one-clock priced timed game
will be a function $v(q,t)$ of the initial state $q$ and the initial
setting $t$ of the single clock. Bouyer {\em et al.}
\cite{BLMR06} showed that the value $v(q,t)$ exists\footnote{Players
  in general cannot guarantee the value exactly, but only
approximate it arbitrarily well -- one of the particular appealing
aspects of SPTGs is that they {\em do} have exactly optimal
strategies! This is in contrast to both the general case and
[0,1]-PTGs without resets.} and that for any state $q$, the value
function $t \rightarrow v(q,t)$ is a piecewise linear function of
$t$. By {\em solving} a game, we mean computing an explicit
description of all these functions (i.e., lists of their line segments). 
From such an object, near-optimal strategies can be synthesized.

\smallskip\noindent{\bf Example.} Figure~\ref{fig:examplesptg} shows an SPTG with $n = 5$ states. Circles
are controlled by Player 1 and squares are controlled by Player
2. States and actions have been annotated with rates and costs. If no
cost is given for an action it has cost zero. The figure also
includes graphs of the value functions. Actions are shown in
black and gray, and an optimal strategy profile is shown along the
$x$-axis of the value functions by using these colors -- more precisely, it is the optimal strategy found by our algorithm. Waiting is
shown as white.

If both players follow the indicated optimal strategies, then the play
that starts with state 3 as the current state at time 0, is as follows:
\begin{enumerate}
\item At state 3 at time 0, Player 1 waits until time $\frac{1}{3}$ and then changes the current state to state 2.
\item At state 2 at time $\frac{1}{3}$, Player 2 waits until time $\frac{2}{3}$ and then changes the current state to state 4.
\item At state 4 at time $\frac{2}{3}$, Player 1 does not wait, but immediately changes the current state to state 3.
\item At state 3 at time $\frac{2}{3}$, Player 1 waits until time 1 and then changes the current state to state 1.
\item At state 1 at time 1, Player 2 can not wait, and immediately changes the current state to state $\perp$, a special state indicating that play has terminated.
\end{enumerate}

Notice that the play waits in state 3 twice. This may seem like a
counter-intuitive property of a play where the players play
optimally. In fact, the game can be generalized to a family, such that
the game with $n$ states has a state that is visited $O(n)$ times in
some optimal play.

\begin{figure}
\center
\includegraphics{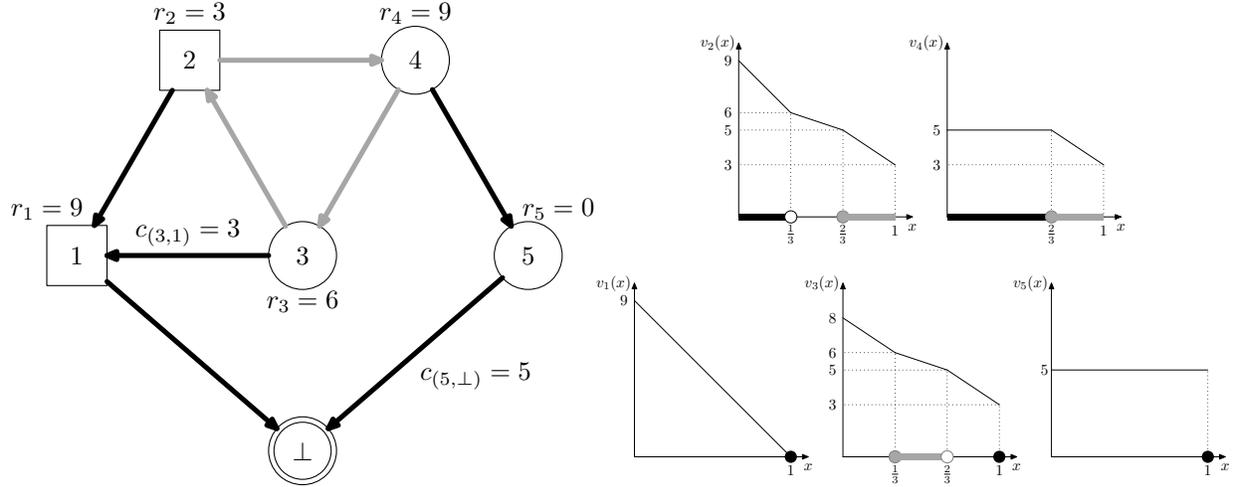}
\caption{Example of an SPTG, showing value
  functions and an optimal strategy profile.}
\label{fig:examplesptg}
\end{figure}

\subsection{Contributions}The contributions of this paper are the following.
\begin{enumerate}
\item{}
\label{cont-red}{\em A polynomial time Turing-reduction from the
problem of solving general PTGs to the problem of solving
SPTGs}. The best previous result along these lines was a
Turing-reduction from the general case to the case of ``[0,1]-PTGs
without resets'' by Bouyer {\em et al.} \cite{BLMR06}. Our reduction
is a polynomial time reduction reducing solving a general PTG to
solving at most $(n+1) (2m+1)$ SPTGs, while the previous reduction is
an exponential time reduction.
\item{}\label{cont-alg}
{\em A novel algorithm for solving SPTGs, based on very different techniques than previously used to solve PTGs}.  In particular, our algorithm is based on applications of a technique from computational geometry: the {\em sweep-line} technique of Shamos and
Hoey \cite{SH76}, applied to the linear arrangement resulting when the
graphs of all value functions are superimposed in a certain way. Also,
an extension of Dijkstra's algorithm due to Khachiyan {\em et al.}
\cite{KBBEGRZ08} is a component of the algorithm. We believe that an implementation of this algorithm and the reduction could provide an attractive alternative to the current state-of-the-art tools for solving PTGs or various special cases (e.g., such as those of UPPAAL, {\tt
  http://uppaal.org} or HyTech {\tt http://embedded.eecs.berkeley.edu/research/hytech/}), which all seem to be based on a value-iteration based algorithm independently devised by Bouyer, Cassez, Fleury, and
Larsen \cite{BCFL04}; and Alur, Bernadsky, and Madhusudan
\cite{ABM04}. We shall refer to that algorithm as the BCFL-ABM algorithm.
\item{}\label{cont-an}{\em A worst case analysis of our algorithm as well as an improved worst case analysis of the BCFL-ABM algorithm.}  
Interestingly, the analysis of
the algorithms is quite indirect: We analyze a different
algorithm for a subproblem (priced games, see section \ref{sec-priced}), namely the 
{\em strategy iteration} algorithm, also used to solve Markov
decision processes and various other classes of two-player zero-sum games
played on graphs, and relate the analysis of this algorithm to our
algorithm. To summarize the result of the analysis, it is convenient
to introduce the parameter $L = L(G)$ of an SPTG to be the total
number of distinct time coordinates of left endpoints of the linear segments of all value
functions of $G$.  Note that the parameter $L$ is very natural, as $L$ is
a lower bound on the size of the explicit description of these value
functions, i.e., the output of the algorithms under consideration. We show:
\begin{enumerate}
\item{}\label{cont-seg}
For an SPTG $G$, we have that $L(G) \leq \min \{12^n, \prod_{k \in S} (|A_k|+1)\}$, where $S$ is the set of states and $A_k$ the set of actions in state $k$. The best previous bound on $L(G)$ was $2^{O(n^2)}$, due to Rutkowski \cite{Rutkowski2011}. 
\item{}\label{cont-worst}The worst case time complexity of our new algorithm is $O((m+n\log n)L)$. In particular, the algorithm combined with the reduction solves general PTGs in time $m 12^{n}n^{O(1)}$. The best previous worst case bound for any algorithm solving PTGs was $2^{O(n^2+m)}$, due to Rutkowski \cite{Rutkowski2011}, who gave this bound for an alternative algorithm, due to him.
\item{}\label{cont-val}The worst case number of iterations of the BCFL-ABM algorithm is  $\min \{12^n, \prod_{k \in S} (|A_k|+1)\} m \cdot n^{O(1)}$ for general PTGs, significantly improving an analysis of Rutkowsi. (An "iteration" is a natural unit of time, specific to the algorithm -- each iteration may take considerable time, as entire graphs of value functions are manipulated during an iteration).  
\item{}\label{cont-reach}For the special case of PTGs with all rates being $1$ (i.e., all
states are equally expensive to wait in) and all transition costs
being $0$ (i.e., Player 1 wants to minimize the time used), our
algorithm combined with the reduction runs in time $O(nm(\min
(m,n^2)+n\log n))$. This special case is also known as {\em timed
  reachability games}, and it was studied by Jurdzinski and Trivedi
\cite{JT07} who gave an exponential
algorithm. Trivedi~\cite{Trivedi09} also observed that the region abstraction algorithm of Laroussinie, Markey, and Schnoebelen \cite{LMS04} can be used to
solve the problem in polynomial time. Our algorithm and reduction
provides an alternative proof of this result.
\item{}\label{cont-aut}For one-clock priced timed automata (the special case of priced timed
games, where all states belong to Player 1), our algorithm combined
with the reduction
runs in time  $O(m n^3(\min (m,n^2)+n\log n))$. This seems to be the best worst case bound known for solving these. It was previously shown to be in NLOGSPACE by Laroussinie, Markey, and Schnoebelen \cite{LMS04}.
\end{enumerate}
\end{enumerate}
The above bounds hold if we assume a unit-cost Real RAM
model of computation, which is a natural model of computation for the algorithms
considered (that previous analyses also seem to have implicitly assumed). The algorithms can also be analyzed in
Boolean models of computation (such as the log cost integer RAM), as a rational valued input yields a
rational valued output. Bounding the bit length of the numbers computed by straightforward inductive techniques, we find that this no more than squares the above worst case complexity bounds. 
The somewhat tedious analysis
establishing this is not included in this version of the paper.
   
\subsection{History of problem and related research}

Priced timed automata (or weighted timed automata) were first
introduced by Alur, Torre, and Pappas~\cite{ATP01} and Behrmann {\em et al.}~\cite{BFHLPRV00}. They showed that
priced timed automata (viewed as one-player games) can be solved in
exponential time. Even before the introduction of priced timed
automata, a special case was studied by Alur and
Dill~\cite{AD94}. They show this case to be PSPACE-hard even for
automata where all states have rate 1 and all actions cost 0. 
Bouyer, Brihaye, Bruyere, and Raskin~\cite{BBBR07} showed that the
problem of solving priced timed automata is in PSPACE. I.e., the
problem is PSPACE-complete when there is no limit on the number of
clocks. Solving one-clock priced timed automata was shown to be in NLOGSPACE by Laroussinie, Markey, and Schnoebelen \cite{LMS04}.

Bouyer, Cassez, Fleury, and Larsen \cite{BCFL04} and Alur, Bernadsky,
and Madhusudan \cite{ABM04} independently introduced the notion of
priced timed games and also both considered value iteration
algorithms for solving priced timed games. 
Even with only 3 clocks, the existence of optimal strategies becomes
undecidable for priced timed games, 
as shown by Bouyer, Brihaye and Markey~\cite{BBM06}. They improved a
similar result of Brihaye, Bruyere, and Raskin~\cite{BBR05}
for 5 clocks.
Hence, various
special cases have been studied.  For timed reachability games,
Jurdzinski and Trivedi \cite{JT07} showed that the problem of computing optimal
values is 
in EXP, and that the problem is EXP-complete for 2 or more clocks. 

For the case with only one clock the problem becomes computable, as
shown by Brihaye, Bruyre, and Raskin~\cite{BBR05}. Bouyer, Larsen,
Markey, and Rasmussen~\cite{BLMR06} gave an explicit triple exponential
time bound on the complexity of solving this problem. This was further
improved to $2^{O(n^2+m)}$ by Rutkowski~\cite{Rutkowski2011}.

\subsection{Organization of paper}
Our algorithm is most naturally presented in three stages, adding more complications to the model at each stage. First, in section \ref{sec-priced}, we show how the strategy iteration paradigm can be used to solve {\em priced games}, where the temporal aspects of the games are not present. In section \ref{sec-simple}, we show how the algorithm extends to simple priced timed games. Finally, in section \ref{ptg}, we show how solving the general case of one-clock priced-timed games can be reduced to the case of simple priced timed games in polynomial time.

In terms of the list of contributions above,
contribution \ref{cont-red}) is Lemma \ref{lem:ptg time}. The
algorithm of contribution \ref{cont-alg}) is $\SolveSPTG$ of Figure
\ref{fig:solve}. Contribution \ref{cont-seg}) is Theorem
\ref{thm:exp}, contribution \ref{cont-worst}) is Theorem
\ref{thm:time},
contribution \ref{cont-val}) is Theorem \ref{thm:bcfl-abm time}, contribution
\ref{cont-reach}) is Theorem \ref{thm:reach} and contribution
\ref{cont-aut}) is Theorem \ref{thm:aut}.
\section{Priced games}
\label{sec-priced}
In this section, we introduce {\em priced games}.
To accommodate lexicographic utilities which will be necessary for subsequent sections, we shall 
consider priced games with utilities in domains other than $\R$.
In this section, we fix any ordered Abelian group $(\Re,+,-,0,\leq)$ for the set of possible utilities. 
We let $\Re_{\geq 0}$ be the set of non-negative elements in $\Re$. In subsequent sections, we will either
have $\Re = {\R}$ or $\Re = {\R} \times{\R}$ with
lexicographic order. 
In the latter case, we write $(x,y)$ as $x + y\epsilon$, where we informally think of 
$\epsilon$ as an infinitesimal. In addition to utilities in the group $\Re$, we also allow the utility $\infty$ (modeling non-termination).

\smallskip\noindent{\bf Formal definition of priced games.} A {\em priced game} $G$ is given by a finite set of {\em states} $S=[n]=
\{1,\ldots,n\}$, a finite set of {\em actions} $A=[m]=\{1,\ldots,m\}$.
The set $S$ is partitioned into $S_1$ and $S_2$, with $S_i$ being the
set of states belonging to Player $i$. 
Player 1 is also referred to as the \emph{minimizer} and Player 2 is
referred to as the \emph{maximizer}.
The set $A$ is
partitioned into $(A_k)_{k \in S}$, with $A_k$ being the set of
actions available in state $k$. Furthermore, define $A^i = \bigcup_{k \in S_i}
A_k$. Each action $j \in A$ has an
associated non-negative {\em cost} $c_j \in \Re_{\geq 0} \cup \{\infty\}$ and an 
associated {\em destination} $\dest(j) \in S \cup \{\perp\}$, where $\perp$
is a special {\em terminal state}. Note that $G$ can be interpreted as
a directed weighted graph. In fact, a priced game can be viewed as a
single source shortest path problem from the point of view of Player
1, with the exception that an adversary, Player 2, controls some of
the decisions.

\smallskip\noindent{\bf Positional strategies.} A {\em positional strategy} for Player~$i$ is a map $\sigma_i$
of $S_i$ to $A$, with $\sigma_i(k) \in A_k$ for each $k \in S_i$.
A pair of strategies (or \emph{strategy profile}) $\sigma =
(\sigma_1,\sigma_2)$ defines a maximal \emph{path} $P_{k_0,\sigma} =
(k_0,k_1,\dots)$, from each $k_0 \in S \cup \{\perp\}$, possibly
ending at $\perp$, such that $\dest(\sigma(k_i))
= k_{i+1}$ for all $i \ge 0$. Note that $\sigma$ can be naturally
interpreted as a map from $S$ to $A$. Let $\ell(k,\sigma)$ be the
length of $P_{k,\sigma}$. The path $P_{k,\sigma}$ recursively defines a {\em
  payoff} $u(k,\sigma) \in \Re \cup \{\infty\}$, paid by Player 1 to Player 2, as follows:
\[
u(k,\sigma) = \begin{cases}
\infty & \text{if $\ell(k,\sigma) = \infty$}\\
0 & \text{if $k = \perp$}\\
c_{\sigma(k)} + u(\dest(\sigma(k)),\sigma) & \text{otherwise}
\end{cases}
\]
I.e., the payoff is the total cost of the path $P_{k,\sigma}$ from $k$
to the terminal state $\perp$, or $\infty$ if $P_{k,\sigma}$ does not
reach $\perp$.

\smallskip\noindent{\bf Values and optimal strategies.} The {\em lower value} $\underline{v}(k)$ of a state $k$ is defined by
$\underline{v}(k) = \max_{\sigma_2} \min_{\sigma_1} u(k,\sigma_1,
\sigma_2)$. A strategy $\sigma_2$ is called {\em optimal}, if for all
states $k$, we have
$\sigma_2 \in \argmax_{\sigma_2} \min_{\sigma_1} u(k,\sigma_1,
\sigma_2)$. Similarly, the {\em upper value} $\overline{v}(k)$ of a state
$k$ is defined by $\overline{v}(k) = \min_{\sigma_1} \max_{\sigma_2}
u(k, \sigma_1, \sigma_2)$ and a strategy $\sigma_1$ is called optimal
if for all $k$, $\sigma_1 \in \argmin_{\sigma_1} \max_{\sigma_2} u(k,
\sigma_1, \sigma_2)$. Khachiyan {\em et al.}~\cite{KBBEGRZ08}
observed that $\underline{v}(k) = \overline{v}(k)$, i.e., that priced
games have {\em values} $v(k) := \underline{v}(k) = \overline{v}(k)$. They
also showed how to find these values and optimal strategies
efficiently using a
variant of Dijkstra's algorithm. The strategies found were postional, hence implying that {\bf optimal positional strategies always exists}.
The $\ExtendedDijkstra$ algorithm is
shown in Figure~\ref{fig:dijkstra}, with $v$ being the vector of
values. Viewing a priced game as a
single source shortest path problem, it is not surprising that it can
be solved by a Dijkstra-like algorithm. Intuitively, if an arc to be taken by Player 2 would be optimal for Player 1, Player 2 will, if possible, do anything else and, informally, ``delete'' the arc.

\smallskip\noindent{\bf Example.} Figure~\ref{fig:examplepricedgame} shows an example of a priced
game. The round vertices are controlled by Player 1, the minimizer,
and the square vertices are controlled by Player 2, the
maximizer. Bold arrows indicate actions used by a strategy profile
$\sigma$, and dashed arrows indicate unused actions. Actions are labeled
by their cost, except if the cost is zero. Finally, the states have
been annotated by the values. Note that $\sigma$ is an optimal
strategy profile.

\begin{figure}
\center
\includegraphics{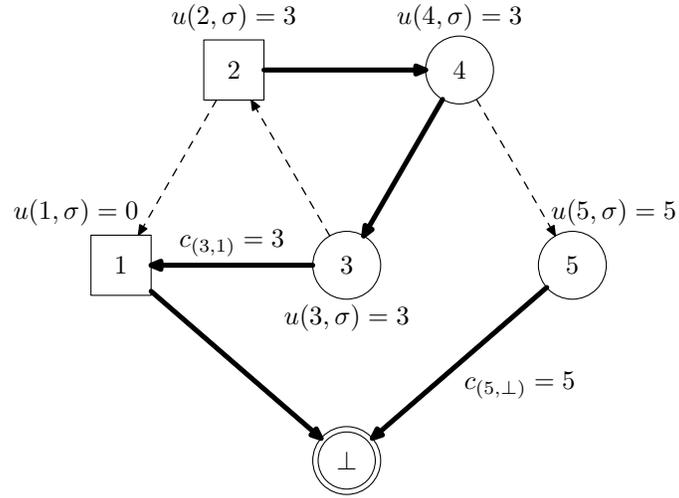}
\caption{Example of a priced game and a strategy profile $\sigma$.}
\label{fig:examplepricedgame}
\end{figure}

\begin{figure}
\begin{center}
\parbox{3.6in}{
\begin{function}[H]

$(v(\perp),v(1),\dots,v(n)) \gets (0,\infty,\dots,\infty)$\;
\While{$S \ne \emptyset$}{
$\displaystyle (k,j) \gets \argmin_{k\in S,j \in A_{k}} c_j + v(d(j))$\;
\If{$k \in S_1$ or $|A_k| = 1$}{
$v(k) \gets c_j + v(d(j))$\;
$\sigma(k) \gets j$\;
$S \gets S \setminus \{k\}$\;
}
\Else{
$A_k \gets A_k \setminus \{j\}$\;
}
}
\Return{$(v,\sigma)$}\;
\caption{\ExtendedDijkstra(\mbox{$G$})}
\end{function}
}
\end{center}
\caption{The $\ExtendedDijkstra$ algorithm of Khachiyan {\em et al.}~\cite{KBBEGRZ08} for
  solving priced games.}
\label{fig:dijkstra}
\end{figure}

\smallskip\noindent{\bf Nash equilibrium.} We say that $\sigma_1$ is a \emph{best response} to $\sigma_2$ if $\sigma_1
\in \argmin_{\sigma_1} u(k,\sigma_1,\sigma_2)$, for all $k \in
S$. Note that $\sigma_1$ is optimal for the game where Player 2 is
restricted to play according to $\sigma_2$.
Similarly, $\sigma_2$ is a \emph{best response} to $\sigma_1$ if
$\sigma_2 \in \argmax_{\sigma_2} u(k,\sigma_1,\sigma_2)$, for all $k
\in S$. A strategy profile $\sigma = (\sigma_1,\sigma_2)$ is a \emph{Nash
equilibrium} if $\sigma_1$ is a best response to $\sigma_2$, and
$\sigma_2$ is a best response to $\sigma_1$.
The following is a standard lemma that establishes the connection between Nash
equilibria and values of zero-sum games.
\begin{lemma}\label{lemma:pgnash}
If $\sigma = (\sigma_1,\sigma_2)$ is a Nash equilibrium, then $v(k) =
u(k,\sigma)$ for all $k \in S$.
\end{lemma}

\begin{proof}
Assume that either $\sigma_1$ or $\sigma_2$ is not optimal. We
will show that $(\sigma_1,\sigma_2)$ is not a Nash equilibrium for
play starting in some state of the game. Assume, without loss of
generality, that $\sigma_1$ does not guarantee Player 1 the payoff
$v(k)$ for play starting at $k$. There are two cases.
\begin{itemize}\item{}Case 1: $u(k, \sigma_1, \sigma_2) \leq v(k)$. In
  this case, Player 2 can deviate from $\sigma_2$ to play a best
  response to $\sigma_1$ at state $k$. Since $\sigma_1$ does, by
  assumption, not guarantee Player 1 $v(k)$, this will yield a larger
  payoff than $v(k)$, i.e., the deviation improves payoff for Player 2
  and $(\sigma_1, \sigma_2)$ is therefore not a Nash equilibrium for
  play starting at $k$. 
\item{}Case 2: $u(k, \sigma_1, \sigma_2) > v(k)$. In this case, Player
  1 can deviate to play an optimal strategy $\sigma_1^*$. By
  definition of optimal, this improves his payoff to $v(k)$ and
  $(\sigma_1, \sigma_2)$ is therefore not a Nash equilibrium for play
  starting at $k$. 
\end{itemize}

\end{proof}

\smallskip\noindent{\bf Strategy iteration algorithm.} We shall present a different algorithm for solving priced games,
following the general {\em strategy iteration} pattern~\cite{Howard60}. This algorithm
will be extended to priced timed games in the next sections. The algorithm presented in the next section can also use the variant of Dijkstra previously described, but our analysis of the running time will use this strategy iteration algorithm in non-trivial ways. Let
$\sigma$ be a strategy profile. For each state $k \in S$, we define
the \emph{valuation} $\nu(k,\sigma) = (u(k,\sigma), \ell(k,\sigma))$. I.e.,
the valuation of a state $k$ for strategy profile $\sigma$ is the
payoff for $k$ combined with the length of the path $P_{k,\sigma}$. If
$\nu(k,\sigma) = (\infty,\infty)$ we write $\nu(k,\sigma) =
\infty$. We say that an action $j \in A_k$ from state $k$ is an
\emph{improving switch for Player 1} if:
\[
(c_j+u(\dest(j),\sigma), 1 + \ell(\dest(j),\sigma)) < \nu(k,\sigma)
\]
Where we order pairs lexicographically, with the first component being
most significant. 
I.e., an improving switch for Player 1 either produces a path from $k$
of smaller cost or with the same cost and smaller length. If a path of
smaller cost is produced we say that $j$ is a \emph{strongly improving switch}.
Similarly, $j \in A_k$ is an improving switch for
Player 2 if:
\[
(c_j+u(\dest(j),\sigma), 1 + \ell(\dest(j),\sigma)) > \nu(k,\sigma)
\]
and $j$ is a strongly improving switch for Player 2 if
$c_j+u(\dest(j),\sigma) > u(k,\sigma)$.

The inclusion of the length in the definition of an improving switch is crucial since it should be an improving switch for Player~2 to go from not using a self-loop of cost 0 to using a self-loop of cost 0 (if Player~2 always uses a self-loop in state $k$, then $\perp$ is never reached and thus Player~1 pays $\infty$ to Player~2. That is the best payoff for Player~2). If the length is not included using a self-loop of cost 0 is never an improving switch for either player.

\begin{lemma}\label{lemma:term}
Let $\sigma = (\sigma_1,\sigma_2)$ be a strategy profile such that for both
players $i$ there are no improving switches in $A^i$. Then $\sigma_1$
and $\sigma_2$ are optimal.
\end{lemma}

\begin{proof}
By Lemma~\ref{lemma:pgnash} it is enough to show that $(\sigma_1,\sigma_2)$ is a Nash
equilibrium for play starting in each state of the game.


Let $\sigma_1'$ be a best response to $\sigma_2$,
and let $\sigma' = (\sigma_1',\sigma_2)$.
Assume, for the sake of contradiction, that there exists a $k_0 \in S$
such that $u(k_0,\sigma') < u(k_0,\sigma)$. Let
$k_i$ be the $i$'th state on the path $P_{k_0,\sigma'}$. I.e.,
$k_{i+1} = \dest(\sigma'(k_i))$.

Either $P_{k_0,\sigma'}$ leads to the terminal state $\perp$, or
$P_{k_0,\sigma'}$ is an infinite path ending in a cycle. The second
case is impossible since that would imply $u(k_0,\sigma') = \infty <
u(k_0,\sigma)$. 

Since $\sigma_1'$ is a best response to $\sigma_2$, we have
$u(k_{i},\sigma') \le u(k_{i},\sigma)$ for all $i$. Also, since $u(\perp,\sigma')
= u(\perp,\sigma) = 0$ and $u(k_0,\sigma') < u(k_0,\sigma)$, when
$P_{k_0,\sigma'}$ leads to the terminal state, there must
exist an index $i$ such that $u(k_{i+1},\sigma')
= u(k_{i+1},\sigma)$ and $u(k_{i},\sigma') < u(k_{i},\sigma)$. Thus,
$\sigma'(k_{i})$ is an improving switch for Player 1, and since
$\sigma(k_{i}) \ne \sigma'(k_{i})$ we have $k_{i} \in S_1$; a
contradiction.

The argument for Player 2 is similar. Let $\sigma_2'$ be a best
response to $\sigma_1$, and $\sigma' =
(\sigma_1,\sigma_2')$. Assume that $u(k_0,\sigma') > u(k_0,\sigma)$ for
some $k_0 \in S$, and let $k_i$ be the $i$'th state along the path
$P_{k_0,\sigma'}$. For the case when $P_{k_0,\sigma'}$ leads to the
terminal state, the argument is the same except with $<$ and $>$
interchanged. 

When $P_{k_0,\sigma'}$ is an infinite path ending in a cycle we must
have $u(k_0,\sigma') = \infty > u(k_0,\sigma)$. I.e., $u(k_i,\sigma)$
is finite for all $i$. Recall that $c_j \ge 0$
for all $j \in A$. For all $k_i \in S_1$, $\sigma(k_i) =
\sigma'(k_i)$, and, hence:
\[
\nu(k_i,\sigma) =
(c_{\sigma(k_i)} +
u(k_{i+1},\sigma),1+\ell(k_{i+1},\sigma)) >
\nu(k_{i+1},\sigma).
\]
On the other hand, for all $k_i \in S_2$, $\sigma'(k_i)$ is not an
improving switch for Player 2, and, hence:
\[
\nu(k_i,\sigma) \ge
(c_{\sigma'(k_i)} +
u(k_{i+1},\sigma),1+\ell(k_{i+1},\sigma)) >
\nu(k_{i+1},\sigma).
\]
Thus, $P_{k_0,\sigma'}$ leads to a cycle, for which the valuations for
$\sigma$ decrease with each step; a contradiction.

\end{proof}

\smallskip\noindent{\bf Improving sets.} Let $B \subseteq A$ be a set of actions such that $|B \cap A_k| \le 1$
for all $k \in S$, and, for $B \cap A_k \ne \emptyset$, let $j(k,B)$ be
the unique action in $B \cap A_k$. Let $\sigma$ be a strategy
profile, and let $\sigma[B]$ be defined as:
\[
\sigma[B](k) := \begin{cases}
j(k,B) & \text{if $B \cap A_k \ne \emptyset$}\\
\sigma(k) & \text{otherwise.}
\end{cases}
\]
If $B = \{j\}$ we also write $\sigma[j]$.
If $j \in A$ is not an improving switch for one player, we say that $j$
is \emph{weakly improving} for the other player.
We say that $B \subseteq A$ is an \emph{improving set for Player $i$}
if there exists an improving switch $j \in B$ for Player $i$, and for all
$j \in B$, $j$ is weakly improving for Player i.

\begin{lemma}\label{lemma:imp}
Let $\sigma = (\sigma_1,\sigma_2)$ be a strategy profile, and let $B
\subseteq A$ be an improving set for Player 1. Then $\nu(k,\sigma[B])
\le \nu(k,\sigma)$ for all $k \in S$, with strict inequality if
$\sigma[B](k_0)$ is an improving switch for Player 1 w.r.t. $\sigma$.
Similarly, if $B$ is an improving set for Player 2, then $\nu(k,\sigma[B])
\ge \nu(k,\sigma)$ for all $k \in S$, with strict inequality if
$\sigma[B](k_0)$ is an improving switch for Player 1
w.r.t. $\sigma$.
\end{lemma}

\begin{proof}
First consider the case where $B$ is an improving set for Player
1. Let $k_0 \in S$. We must show that $\nu(k_0,\sigma[B]) \le
\nu(k_0,\sigma)$ with strict inequality if $\sigma[B](k_0)$ is an
improving switch for Player 1 w.r.t. $\sigma$. This is clearly true if
$\nu(k_0,\sigma) = \infty$. Thus, assume that $\nu(k_0,\sigma) <
\infty$.

Let $k_i$ be the $i$'th state on the path
$P_{k_0,\sigma[B]}$. Since $\sigma[B](k_i)$ is weakly improving for
Player 1 we have, for all $i$:
\begin{equation}\label{eqt1}
(c_{\sigma[B](k_i)} +
u(k_{i+1},\sigma),1+\ell(k_{i+1},\sigma)) \le \nu(k_i,\sigma)
\end{equation}
with strict inequality exactly when $\sigma[B](k_i)$ is an improving
switch for Player 1 w.r.t. $\sigma$.

From (\ref{eqt1}), and the fact that $c_j \ge 0$ for all $j \in A$, we
get that:
\begin{align*}
\nu(k_i,\sigma) &\ge
(c_{\sigma[B](k_i)} +
u(k_{i+1},\sigma),1+\ell(k_{i+1}),\sigma))\\
&>(u(k_{i+1},\sigma),\ell(k_{i+1},\sigma))\\
&=\nu(k_{i+1},\sigma).
\end{align*}
Hence, $P_{k_0,\sigma[B]}$ does not lead to a cycle, since the
valuations in $\sigma$ can not strictly decrease along the entire cycle.

We next show, using backwards induction on $i$, that
$\nu(k_i,\sigma[B])\leq \nu(k_i,\sigma)$. For the base case, $k_i =
\perp$, the statement is clearly true.
Otherwise, for $i<\ell(k_0,\sigma[B])$, we get from (\ref{eqt1})
and the induction hypothesis that:
\begin{align*}\nu(k_i,\sigma) & \geq(c_{\sigma[B](k_i)} +
u(k_{i+1},\sigma),1+\ell(k_{i+1},\sigma))\\ 
& \geq (c_{\sigma[B](k_i)} +
u(k_{i+1},\sigma[B]),1+\ell(k_{i+1},\sigma[B]))\\ 
& =\nu(k_i,\sigma[B]).
\end{align*}
Note that if $j \in A_{k_i}\cap B$ is an improving switch for
Player 1 then the first inequality is strict.

The proof for the second case, where $B$ is an improving set for
Player 2, is similar. Let $k_0\in S$. We show that $\nu(k_0,\sigma[B]) \ge
\nu(k_0,\sigma)$ with strict inequality if $\sigma[B](k_0)$ is an
improving switch for Player 2 w.r.t. $\sigma$. Now, this is
clearly true if $\nu(k_0,\sigma[B]) = \infty$. If $\nu(k_0,\sigma[B])
< \infty$, it immediately follows that $P_{k_0,\sigma[B]}$ is
of finite length. The rest of the proof is identical, but with $<$ and
$>$ interchanged.

\end{proof}

The following corollary is an important consequence of Lemma
\ref{lemma:imp}.

\begin{corollary}\label{cor:imp}
Let $\sigma = (\sigma_1,\sigma_2)$ be a strategy profile for which
Player 2 has no improving switches. Let $B \subseteq A^1$ be an improving set
for Player 1, and let $\sigma' = (\sigma_1[B],\sigma_2')$ where
$\sigma_2'$ is any strategy for Player 2. Then
$\nu(k,\sigma') \le \nu(k,\sigma)$ for all states $k$, with strict
inequality if $\sigma'(k)$ is an improving switch for Player 1.
\end{corollary}

\begin{proof}
Since Player 2 has no improving switches w.r.t. $\sigma$, every action
$j \in A^2$ emanating from a state controlled by Player 2 is weakly
improving for Player 1 w.r.t. $\sigma$. It follows that $B \cup
\sigma_2'$ is an improving set for Player 1, and we then know from
Lemma \ref{lemma:imp} that $\nu(k,\sigma') \le \nu(k,\sigma)$ for all
states $k$, with strict inequality if $\sigma'(k)$ is an improving
switch for Player 1.

\end{proof}

Lemma~\ref{lemma:term} and Corollary \ref{cor:imp} allow us to define the
$\StrategyIteration$ algorithm as shown in
Figure~\ref{fig:strategyiteration}. $u(\sigma)$ is the vector of
payoffs for $\sigma$.
The algorithm is a local search
algorithm, and Lemma~\ref{lemma:term} ensures that a local optimum is
also a global optimum. Player 1 repeatedly performs improving switches
while Player 2 always plays a best response to the current strategy of
Player 1. Corollary~\ref{cor:imp} is used to prove termination of the
algorithm: it ensures that each strategy is only encountered
once, and the number of strategies is finite.

\begin{figure}
\begin{center}
\parbox{4in}{
\begin{function}[H]

\While{$\exists$ improving set $B_1 \subseteq A^1$ for Player 1 w.r.t. $\sigma$}{
$\sigma \gets \sigma[B_1]$\;
\While{$\exists$ improving set $B_2 \subseteq A^2$ for Player 2 w.r.t. $\sigma$}{
  $\sigma \gets \sigma[B_2]$\;
}
}
\Return{$(u(\sigma),\sigma)$}\;
\caption{\StrategyIteration(\mbox{$G,\sigma$})}
\end{function}
}
\end{center}
\caption{The $\StrategyIteration$ algorithm for solving priced games.}
\label{fig:strategyiteration}
\end{figure}

\begin{theorem}\label{thm:strategyiteration}
The $\StrategyIteration$ algorithm correctly computes an optimal
strategy profile $\sigma^*$ such that neither player has an improving
switch w.r.t. $\sigma^*$.
\end{theorem}

\begin{proof}
It immediately follows from Lemma~\ref{lemma:term} that if the
$\StrategyIteration$ algorithm terminates, it correctly computes an
optimal strategy profile. Indeed, in order to escape both while-loops
neither player $i$ can have an improving switch in $A^i$
w.r.t. $\sigma$.

Let $\sigma = (\sigma_1, \sigma_2)$ be the current strategy profile at
the beginning of the outer while-loop, and let $\sigma[B_1] =
(\sigma_1', \sigma_2)$. From Lemma~\ref{lemma:imp} we know that with
each iteration of the inner while-loop the valuations are
non-decreasing, with at least one state strictly increasing its
valuation. Since there are only finitely many strategies, it follows
that the inner while-loop always terminates. Let the resulting strategy
profile be $\sigma' = (\sigma_1', \sigma_2')$. Observe that Player 2
has no improving switches in $A^2$ w.r.t. $\sigma'$. 

After the first iteration Player 2 has no improving switches
w.r.t. $\sigma$, and it follows 
from Corollary \ref{cor:imp} that the valuations are non-increasing
and strictly decreasing for at least one state from $\sigma$ to $\sigma'$.
Again, since there
are only finitely many strategies the outer while-loop is guaranteed
to terminate.

\end{proof}

\section{Simple priced timed games}
\label{sec-simple}
A \emph{simple priced timed game} (SPTG) $G$ is given by a priced game
$G' = (S_1,S_2,(A_k)_{k \in S}, (c_j)_{j \in A}, d)$, where $S = S_1
\cup S_2$ and $A = \bigcup_{k \in S} A_k$, and for each state $i \in
S$, an associated \emph{rate} $r_i \in \R_{\geq
  0}$. We assume that $A_k \ne \emptyset$ for all $k \in S$.

\smallskip\noindent{\bf Playing an SPTG.} A SPTG $G$ is played as follows. A pebble is placed on some starting
state $k_0$ and the clock is set to its starting time
$x_0$. The pebble is then moved from state to state by the
players. The current configuration of the game is described
by a state and a time, forming a pair $(k, x) \in S \times [0,1]$.

Assume that after $t$ steps the pebble is on state $k_t \in S_i$,
controlled by Player $i$, at time $x_t$, corresponding to the
configuration $(k_t,x_t)$. Player $i$ now chooses the next action
$j_{t} \in A_{k_t}$. Furthermore, the player also
chooses a delay $\delta_{t} \ge 0$ such that $x_{t+1} = x_t +
\delta_{t} \le 1$. 
The pebble is moved to $d(j_{t}) = k_{t+1}$. The
next configuration is then $(k_{t+1},x_{t+1})$. We write 
\[
(k_{t},x_t)
\xrightarrow{j_{t},\delta_{t}} (k_{t+1},x_{t+1}).
\]
The game ends if $k_{t+1} =\ \perp$.

\smallskip\noindent{\bf Plays and outcomes.} A \emph{play} of the game is a sequence of steps starting from some
configuration $(k_0,x_0)$.
Let 
\[
\rho = (k_0,x_0) \xrightarrow{j_0, \delta_0} (k_1,x_1)
\xrightarrow{j_1, \delta_1} \dots \xrightarrow{j_{t-1}, \delta_{t-1}}
(k_{t},x_{t})
\]
be a finite play such that $k_{t} =\ 
\perp$. The outcome of the game, paid by Player 1 to Player 2, is then
given by:
\[
\cost(\rho) = \sum_{\ell=0}^{t-1} (\delta_\ell r_{k_\ell} + c_{j_\ell}).
\]
I.e., for each unit of time spent waiting at a state $k$ 
Player 1 pays the rate $r_k$ to Player 2. Furthermore, every time an
action $j$ is used, Player 1 pays the cost $c_j$ to Player 2. 
If $\rho$
is an infinite play the outcome is $\infty$, and we write $\cost(\rho) =
\infty$.

\smallskip\noindent{\bf Positional strategies.} A (positional) strategy for Player $i$ is a map $\pi_i: S_i \times [0,1]
\to A \cup \{\lambda\}$, where $\lambda$ is a special delay action.
For every $k \in S_i$ and $x \in [0,1)$, if $\pi_i(k,x) = \lambda$
then we require that there exists a $\delta > 0$ such that for all $0
\le \epsilon < \delta$, $\pi_i(k,x+\epsilon) = \lambda$.
Let $\delta_{\pi_i}(k,x) = \inf \{x'-x \mid x \le x' \le 1, \pi_i(k,x') \ne
\lambda\}$ be the delay before the pebble is moved when starting in
state $k$ at time $x$ for some strategy $\pi_i$. More general types of strategies could be considered, but see Remark~\ref{rem:common positional strategies}.

\smallskip\noindent{\bf Playing according to a strategy and strategy profiles.} Player $i$ is said to play according to $\pi_i$ if, when the pebble is
in state $k \in S_i$ at time $x \in [0,1]$, he waits until time $x' =
x + \delta_{\pi_i}(k,x)$ and then moves according to $\pi_i(k,x')$. 
A \emph{strategy profile} $\pi = (\pi_1,\pi_2)$
is a pair of strategies, one for each player. Let $\Pi_i$ be the set
of strategies for Player $i$, and let $\Pi$ be the set of all strategy
profiles. A strategy profile $\pi$ is again interpreted as a map
$\pi: S \times [0,1] \to A \cup \{\lambda\}$. Furthermore, we use
$\pi(x)$ to refer to the decisions at a fixed time. I.e., $\pi(x): S
\to A \cup \{\lambda\}$ is the map defined by $(\pi(x))(k) =
\pi(k,x)$.

\smallskip\noindent{\bf Value functions and optimal strategies.} Let $\rho_{k,x}^{\pi}$ be the play starting from
configuration $(k,x)$ where the players play according to $\pi$.
Define the \emph{value function} for a strategy
profile $\pi = (\pi_1,\pi_2)$ and state $k$ as:
$
v_k^{\pi_1,\pi_2}(x) = \cost(\rho_{k,x}^{\pi})
$.
For fixed strategies $\pi_1$ and $\pi_2$ for Player 1 and 2, define
the \emph{best response value functions} for Player 2 and 1,
respectively, for a state $k$ as:
\begin{align*}
v_k^{\pi_1}(x) &= \sup_{\pi_2 \in \Pi_2} ~ v_k^{\pi_1, \pi_2}(x) \\
v_k^{\pi_2}(x) &= \inf_{\pi_1 \in \Pi_1} ~ v_k^{\pi_1, \pi_2}(x)
\end{align*}
We again define \emph{lower} and \emph{upper value functions}: 
\begin{align*}
\underline{v}_k(x) &= \sup_{\pi_2 \in \Pi_2} ~ v_k^{\pi_2}(x) = \sup_{\pi_2 \in \Pi_2} ~ \inf_{\pi_1 \in \Pi_1}
~ v_k^{\pi_1, \pi_2}(x).\\
\overline{v}_k(x) &= \inf_{\pi_1 \in \Pi_1} ~ v_k^{\pi_1}(x) = \inf_{\pi_1 \in \Pi_1} ~ \sup_{\pi_2 \in \Pi_2} ~
v_k^{\pi_1, \pi_2}(x).
\end{align*}
Note that $\inf$ and $\sup$ are used because there are infinitely many
strategies. Bouyer \emph{et al.}~\cite{BLMR06} showed that
$\underline{v}_k(x) =
\overline{v}_k(x)$. In fact, this was shown for the more general class
of \emph{priced timed games} (PTGs) studied in Section~\ref{ptg}. Thus, every
SPTG has a \emph{value function} $v_k(x) :=
\underline{v}_k(x) = \overline{v}_k(x)$ for each state $k$.

\begin{remark}{\em \label{rem:common positional strategies}
Let us remark that positional strategies are commonly defined as maps from states
and times to delays and actions. For instance, positional strategies {\em (CDPS)} are commonly defined as $\tau_i: S_i \times [0,1]
\to [0,1] \times A$. This is more general than
our definition of strategies, since $\tau_i(k,x) = (\delta,a)$
with $\delta > 0$ does not imply that for all $x' \in (x,x+\delta]$ we
have $\tau_i(k,x') = (x+\delta-x',a)$, whereas this implication holds for the strategies we
use. We choose to use the specialized definition of strategies because
it offers a better intuition for understanding the proposed
algorithm. It is easy to see that the players can
not achieve better values by using the more general strategies.
Indeed, let $\tau_i$ be some strategy where $\tau_i(k,x) =
(\delta,a)$ and $\tau_i(k,x') = (\delta',a')$, such that $[x,x+\delta] \cap [x',x'+\delta'] \ne \emptyset$. Then
one of the following two modifications will not make $\tau_i$ achive worse values: $\tau_i(k,x) = (x'+\delta'-x,a')$ or $\tau_i(k,x') = (x+\delta-x',a)$.
Even more general strategies can be considered that depends on the history of the play so far in arbitary ways, but as shown by Bouyer \emph{et al.}~\cite{BLMR06} for any $\epsilon>0$ there exists a CDPS $\pi_1$ for Player~$1$, such that against any strategy $\pi_2$ for Player~2 $v_k^{\pi_1, \pi_2}(x)\geq v_k(x)-\epsilon$ (and similar if Player~2 must play a CDPS). Thus the same is the case for our more specialised positional strategies. That we only consider positional strategies is thus not a restriction.
}
\end{remark}

\smallskip\noindent{\bf Strategies optimal from some time.} A strategy $\pi_i \in \Pi_i$ is \emph{optimal from time $x$} for
Player $i$ if:
\[
\forall k \in S,~ x' \in [x,1]: \quad
v_k^{\pi_i}(x') = v_k(x').
\]
Strategies are called \emph{optimal} if they are optimal from
time 0. Similarly, a strategy $\pi_i$ is a \emph{best response} to
another strategy $\pi_{-i}$ from time $x$ if:
\[
\forall k \in S,~ x' \in [x,1]: \quad
v_k^{\pi_i,\pi_{-i}}(x') = v_k^{\pi_{-i}}(x').
\]

%

\smallskip\noindent{\bf Nash equilibrium
from some time.} A strategy profile $(\pi_1,\pi_2)$ is called a \emph{Nash equilibrium
from time $x$} if $\pi_1$ is a best response to $\pi_2$ from time $x$,
and $\pi_2$ is a best response to $\pi_1$ from time $x$. As in the case of
Lemma~\ref{lemma:pgnash} for priced games, any equilibrium payoff of an SPTG is the value of the game. 
The exact statement is shown in
Lemma~\ref{lemma:nash}. Since the argument is standard, and similar to
the proof of Lemma~\ref{lemma:pgnash}, it has been omitted. Just note
that instead of considering best responses, which we have not yet
showed exist for SPTGs, it suffices to use some better strategy.

\begin{lemma}\label{lemma:nash}
If there exists a strategy profile $(\pi_1,\pi_2)$ that is a Nash
equilibrium from time $x$, then $v_k(x') = v_k^{\pi_1,\pi_2}(x')$ for
all $k\in S$ and $x' \in [x,1]$.
\end{lemma}

The existence of optimal strategies and best replies is
non-trivial. We are, however, later going to prove the following
theorem, which, in particular, implies that $\inf$ and $\sup$ can be
replaced by $\min$ and $\max$ in the definitions of value functions.
(That value functions are piecewise linear also holds for general PTGs and was
first established by Bouyer {\em et al.} \cite{BLMR06} who furthermore showed
that the value functions are not continuous in general and that there are PTGs for which no optimal strategy exists, thus showing that we can not in general replace $\inf$ and $\sup$ by $\min$ and $\max$.)
\begin{theorem}\label{thm:opt}
For any SPTG there exists an optimal strategy profile. Also, 
the value functions are continuous piecewise linear functions.
\end{theorem}
Our proof will be algorithmic. Specifically, the algorithm
$\SolveSPTG$ computes a value function of the desired kind.
Furthermore, the proof of correctness of $\SolveSPTG$ 
(the proof of Theorem \ref{thm:time}) also
yields the existence of exactly optimal strategies.

We refer to the non-differentiable points of the value functions of $G$
as \emph{event points} of $G$. The number of distinct event points of $G$
is an important parameter in the
complexity of our algorithm for solving SPTGs. We denote by
$\EventPoints(G)$ the total number of event points, excluding $x =
1$. 

\subsection{Solving SPTGs\label{sec:solve sptgs}}

In order to solve an SPTG we make use of a technique similar to the
\emph{sweep-line} technique from computational geometry of Shamos and Hoey
\cite{SH76}. Informally, we construct the value functions by
moving a sweep-line backwards from time 1 to time 0, and at each time
computing the current values based on the later values. The approach
is also similar to a technique known in game theory
as {\em backward induction}. The parameter of the induction, the time,
is  a continuous parameter, however. The BCFL-ABM
algorithm also applies backward
induction, but there the parameter of induction is the number of
transitions taken, i.e., a discrete parameter, leading to a {\em
value iteration} algorithm. 

\smallskip\noindent{\bf Informal description of the algorithm.} If $\pi$
is a strategy profile that is optimal from time $x$ in an SPTG $G$, we 
use $\pi$ to construct a new strategy profile $\pi'$ that is optimal
from time $x' < x$. More precisely, for $x'$ sufficiently
close to $x$, we show that there exists a fixed optimal action  (where
"waiting" is viewed as an action)  for all states
for both players for every point in time in the interval $[x',x)$.
The new strategy profile $\pi'$ is then obtained from $\pi$ by using these
actions. Starting from time $x'$, once the players wait in some state
$k$, they wait at least until time $x$ because they use the same
actions throughout the interval. This allows us to model the situation
with a priced game where every state $k$ is given an additional action
$\lambda_k$ corresponding to waiting for $y = x-x'$ units of time.
Thus, the value of a state in the priced game is the same as the value
of the corresponding state in $G$ if the game starts at time $x' = x-y$,
and if the first time a player waits he is forced to wait until time
$x$. 
The formal development of the algorithm follows. The following
definition is a formalization of the priced game described above.

\begin{definition}\label{def:Gxe}
For a given SPTG $G = (S_1,S_2,(A_k)_{k \in S}, c, d, r)$, a time
$x \in (0,1]$, and $y \ge 0$, let the priced game $G^{x,y} =
(S_1,S_2,(A_k')_{k \in S}, c^{x,y}, d')$ be defined by:
\begin{alignat*}{2}
\forall k \in S&:&\quad A_k' &= A_k \cup
\{\lambda_k\}\\
\forall j \in A_k'&:&\quad
c^{x,y}_j &= \begin{cases}
v_k(x) + y r_k & \text{if $j = \lambda_k$}\\
c_j & \text{otherwise}\\
\end{cases}\\
\forall j \in A_k'&:&\quad
d'(j) &= \begin{cases}
\perp & \text{if $j = \lambda_k$}\\
d(j) & \text{otherwise}\\
\end{cases}
\end{alignat*}
We refer to actions $\lambda_k$, for $k \in S$, as \emph{waiting actions}.
\end{definition}

We sometimes write $u(k,\sigma,G^{x,y})$ instead of
$u(k,\sigma)$ to clarify which priced game $G^{x,y}$ we
consider.

\smallskip\noindent{\bf The rates obtained in $G^{x,y}$.} Let $x
\in (0,1]$ and $y \ge 0$. Let $\sigma$ be a strategy profile for
$G^{x,y}$, and let $k_0$ be a state. Consider the (maximal) path
$P_{k_0,\sigma} = (k_0,k_1,\dots)$ that starts at $k_0$ and uses
actions of $\sigma$. When $P_{k_0,\sigma}$ is finite and the last
action of $P_{k_0,\sigma}$ is a waiting action $\lambda_k$ for some
$k$, it is useful to introduce notation for refering to the rate of
the second last state of $P_{k_0,\sigma}$. (Note that the last state
of $P_{k_0,\sigma}$ is the terminal state $\perp$.) For this purpose
we define:
\[
r(k_0,\sigma) ~=~ \begin{cases}
r_{k_{t-1}} & \text{if $P_{k_0,\sigma} = (k_0,k_1,\dots,k_t)$ and
  $\sigma(k_{t-1}) = \lambda_{k_{t-1}}$} \\
0 & \text{otherwise}
\end{cases}
\]
Note that for finite paths that move to $\perp$ without using waiting
actions we may interpret the rate of $\perp$ as being 0.
We again sometimes write $r(k,\sigma,G^{x,y})$ instead of
$r(k,\sigma)$ to clarify which priced game $G^{x,y}$ we
consider.
Note that the only differences between $G^{x,y}$ and $G^{x,y'}$ for $y
\ne y'$ are the costs of the waiting actions. Also note that
$r(k_0,\sigma)$ does not depend on $y$. In particular, we always have:
\begin{equation}\label{eq:G^{x,y}}
u(k_0,\sigma,G^{x,y}) = u(k_0,\sigma,G^{x,0}) + y r(k_0,\sigma) ~.
\end{equation}

\smallskip\noindent{\bf The game $G^x$.} We will often let $y$ be the infinitesimal $\epsilon$, in which case
we simply denote $G^{x,\epsilon}$ by $G^x$ and $c^{x,\epsilon}$ by $c^x$.
Since $\epsilon$ is an infinitesimal, the payoffs of a
strategy profile $\sigma$ for $G^x$ have two components.
From (\ref{eq:G^{x,y}}) we (informally) know that 
$u(k_0,\sigma,G^{x}) = u(k_0,\sigma,G^{x,0}) + \epsilon
r(k_0,\sigma)$. There are no infinitesimals in $G^{x,0}$, and, hence, the
second component of the payoff $u(k_0,\sigma,G^{x})$ is exactly
$r(k_0,\sigma)$.
For every $x \in (0,1]$ we let $\sigma^x = (\sigma_1^x,\sigma_2^x)$ be
a strategy profile for which neither player has an improving
switch. I.e., by Lemma~\ref{lemma:term} $\sigma^x$ is an optimal
strategy profile for $G^x$. The existence of $\sigma^x$ is guaranteed
by Theorem \ref{thm:strategyiteration}.
To shorten notation we let $a^x(k) = u(k,\sigma^x,G^{x,0})$ and
$b^x(k) = r(k,\sigma^x,G^{x})$.

\begin{lemma}\label{lemma:zero}
The strategy profile $\sigma^x$ is optimal for $G^{x,0}$ and $a^x(k) =
v_k(x)$.
\end{lemma}

\begin{proof}
To prove the first part of the lemma we observe that if a better value
can be achieved in $G^{x,0}$ then, regardless of the infinitesimal
component of the payoff achieved by
$\sigma^x$ in $G^x$, it will also be better for $G^x$. Furthermore, the value
of a state $k$ in $G^{x,0}$ must be consistent with $v_k(x)$. It
follows that $a^x(k) = u(k,\sigma^x,G^{x,0}) = v_k(x)$.
\end{proof}



Note that the only difference between $G^x$ and $G^{x'}$,
for $x \ne x'$, is the costs of the waiting actions $\lambda_k$. Hence,
we may interpret a strategy profile $\sigma$ for $G^x$ as a
strategy profile for $G^{x'}$.
Also note that $G^{x}$ is identical to the priced game $G'$
defining $G$, except that for each state $k$ there is an additional
action $\lambda_k$ corresponding to waiting in that state in the SPTG.
Slightly abusing notation, we will interpret actions chosen by $\sigma$ as also being actions
$\pi(x)$ for $G$, and the actions of $\pi(x)$ as forming a strategy profile for $G^x$.

The following lemma establishes the connection between the SPTG $G$
and the priced game $G^x$, for some $x \in (0,1]$.

\begin{lemma}\label{lemma:connection}
Let $\pi$ be a strategy profile for $G$ that is optimal from time $x$,
and let $x' < x$. If $\pi(x'') = \sigma^x$ for
all $x'' \in [x',x)$, then $v_k^\pi(x') = v_k(x) +
  (x-x')b^x(k)$
for all $k \in S$.
\end{lemma}

\begin{proof}
Let $\rho_{k,x'}^\pi = (k_0,x_0) \xrightarrow{j_0, \delta_0} (k_1,x_1)
\xrightarrow{j_1, \delta_1} \dots$, and let $t$ be the maximum index
such that $x_t < x$. Since $\pi(x'') = \sigma^x$ for all
$x'' \in [x',x)$, we have $\delta_\ell = 0$ for all $\ell < t$ and
$\delta_t \ge x-x'$. By splitting the cost of $\rho_{k,x'}^\pi$
into cost accumulated before and after time $x$, we get:
\begin{align*}
v_k^\pi(x') &= \cost(\rho_{k,x'}^\pi) \\
&= \left((x-x') r_{k_t} + \sum_{\ell = 0}^{t-1} c_{j_\ell}\right) + v_{k_t}(x) \\
&= u(k,\sigma^x,G^{x,0}) + (x-x')r(k,\sigma^x,G^{x}) \\
&= a^x(k) + (x-x')b^x(k) \\
&= v_k(x) + (x-x')b^x(k) \;.
\end{align*}
The last equality follows from Lemma~\ref{lemma:zero}.
\end{proof}

\smallskip\noindent{\bf The function $\NextEventPoint(G^x)$.} 
For every action $j \in A$ and time $x \in (0,1]$, define the
function:
\begin{align*}
f_{j,x}(x'') ~:\!&=~ c_j + u(\dest(j),\sigma^x,G^{x,x-x''}) \\
~&=~ c_j + a^x(\dest(j)) + (x - x'') b^x(\dest(j))~.
\end{align*}
I.e., $f_{j,x}(x'')$, for $j \in A_k$, is the payoff obtained in
$G^{x,x-x''}$ by starting at state $k$, using action $j$, and then
repeatedly using actions of $\sigma^x$. In particular, we have
$u(k,\sigma^x,G^{x,x-x''}) = f_{\sigma^x(k),x}(x'')$, and $j \in A_k$,
for $k \in S_1$, is a strongly improving switch for Player 1 w.r.t. $\sigma^x$
if and only if $f_{j,x}(x'') < f_{\sigma^x(k),x}(x'')$. A similar observation can be
made for Player 2. Note that $f_{j,x}(x'')$ defines a line in the
plane.

From the definition of $\sigma^x$ we know that $\sigma^x$ is optimal
for $G^{x,x-x''}$ when $x''$, with $x'' < x$, is sufficiently close to
$x$. In the following we will be interested in making $x''$ as small
as possible
while maintaining that $\sigma^x$ is optimal for $G^{x,x-x''}$. Recall
that $\sigma^x$ is optimal for $G^{x,x-x''}$ if neither
player has an improving switch w.r.t. $\sigma^x$. 
Also, there are no improving switches
when $y = x-x'' > 0$ is sufficiently small. Recall that the valuation
$\nu(k,\sigma^x,G^{x,y}) =
(u(k,\sigma^x,G^{x,y}),\ell(k,\sigma^x,G^{x,y}))$ consists of two
components. Although the payoffs $u(k,\sigma^x,G^{x,y})$ change for
different $y$, the path lengths $\ell(k,\sigma^x,G^{x,y})$ remain the 
same. Hence, there are no improving switches w.r.t. $\sigma^x$
for $G^{x,x-x''}$ if and only if:
\begin{align}
&\forall k \in S_1, j \in A_k: \quad (f_{\sigma^x(k),x}(x''),\ell(k,\sigma^x,G^{x,0})) ~\ge~
(f_{j,x}(x''), 1+\ell(d(j),\sigma^x,G^{x,0})) \label{ineq:im1}\\
&\forall k \in S_2, j \in A_k: \quad (f_{\sigma^x(k),x}(x''),\ell(k,\sigma^x,G^{x,0})) ~\le~
(f_{j,x}(x''), 1+\ell(d(j),\sigma^x,G^{x,0})) \label{ineq:im2}
\end{align}
where the pairs are compared lexicographically.

Let $x'$ be the first intersection before $x$ of two lines
$f_{\sigma^x(k),x}(x'')$ and $f_{j,x}(x'')$, for $k \in S$ and $j \in
A_k \setminus \sigma^x$. Then (\ref{ineq:im1}) and (\ref{ineq:im2})
are satisfied for all $x'' \in (x',x)$, since (\ref{ineq:im1}) and
(\ref{ineq:im2}) are satisfied for $x''$ sufficiently close to $x$,
and the relations (inequalities) between $f_{\sigma^x(k),x}(x'')$ and
$f_{j,x}(x'')$ remain the same for all $k \in S$ and $j \in A_k$. In
particular, neither player has a strongly improving switch
w.r.t. $\sigma^x$ for $G^{x,x-x''}$ when $x'' \in [x' x)$. On
the other hand, either (\ref{ineq:im1}) or (\ref{ineq:im2}) is no longer
satisfied when $x'' < x'$. Hence, $\sigma^x$ is optimal for
$G^{x,x-x''}$ when $x' < x'' < x$, but not when $x'' < x'$. Note also
that if $\sigma^x$ is optimal for $G^{x,x-x''}$ for all $x'' \in
(x',x)$, then $\sigma^x$ is also optimal for $G^{x,x-x'}$. Indeed,
assume for the sake of contradiction that $\sigma^x$ is not optimal
for $G^{x,x-x'}$, then when $x''$ is sufficiently close to $x'$ it
must also be possible for one of the players to improve his value in 
$G^{x,x-x''}$ by switching to an optimal strategy for $G^{x,x-x'}$; a
contradiction. Hence, we have shown that $x'' = x'$ is the smallest
time for which $\sigma^x$ remains optimal for $G^{x,x-x''}$.
Note that although $\sigma^x$ is optimal for
$G^{x,x-x'}$ it is still possible that some player has an improving
switch that improves the path lengths without changing the payoffs.

We define $\NextEventPoint(G^x)$ to be the time
$x'$ of the first intersection before $x$ of two lines
$f_{\sigma^x(k),x}(x'')$ and $f_{j,x}(x'')$, for $k \in S$ and $j \in
A_k \setminus \sigma^x$.
We will later see that $x'$ is the next event
point preceding $x$, which justifies the name.
We next derive a concise definition of $\NextEventPoint(G^x)$.
Note that $x'$ must satisfy $f_{j,x}(x') =
f_{\pi(k),x}(x')$ and $f_{j,x}(x) \ne f_{\pi(k),x}(x)$ for some 
$k \in S$ and $j \in A_k$. I.e., if  
$f_{j,x}(x'') = f_{\sigma^x(k),x}(x'')$ for all $x''$, then $j$ can
never be an improving switch w.r.t. $\sigma^x$, and we may ignore the
action $j$ when looking for $x'$. Since we are working with lines it
is enough to check two different points to see whether two lines are
the same. It follows that
$\NextEventPoint(G^x)$ can be defined as:
\[
\max~\{0\}  \cup \{x' \in [0,x) \mid \exists k \in S, j \in
  A_k: ~
f_{j,x}(x') = f_{\pi(k),x}(x') \wedge f_{j,x}(x) \ne f_{\pi(k),x}(x) \}.
\]
Note that $\NextEventPoint(G^x)$ is well-defined, since there is only
one function $f_{j,x}$ for each action $j \in A$. 

The following lemma follows as a consequence of the discussion above.

\begin{lemma}\label{lemma:interval}
Let $x' = \NextEventPoint(G^x)$, then $\sigma^x$ is optimal for
$G^{x,y}$, for all $y \in (0,x-x']$. Furthermore, neither player has a
strongly improving switch w.r.t. $\sigma^x$ for $G^{x,y}$.
\end{lemma}

We are now ready to state the main technical lemma used to prove
the correctness of our algorithm.

\begin{lemma}\label{lemma:induction}
Let $x' = \NextEventPoint(G^x)$, and let $\pi = (\pi_1,\pi_2)$ be a
strategy profile that is optimal from time $x$. 
Then the strategy profile $\pi' = (\pi'_1,\pi'_2)$, defined by:
\[
\pi'(k,x'') = \begin{cases}
\sigma^x(k) & \text{if $x'' \in [x',x)$}\\
\pi(k,x'') & \text{otherwise}
\end{cases}
\]
is optimal from time $x'$, and $v_k(x'') = v_k(x) + b^x(k) (x-x'')$,
for $x'' \in [x',x)$ and $k\in S$.
\end{lemma}

\begin{proof}
Let us first note that for any strategy profile $\pi''$, the outcome
$v_{k_0}^{\pi''}(x_0)$, for some starting configuration $(k_0,x_0)$,
only depends on the choices made by $\pi''$ in the interval
$[x_0,1]$. Hence, since $\pi'$ is the same as $\pi$ in the interval
$[x,1]$, $\pi'$ is also optimal from time $x$.

Let us also note that $v_k(x'') = \infty$ for some $k \in S$ and $x''
\in [0,1]$ if and only if $v_k(x'') = \infty$ for all $x''
\in [0,1]$. Indeed, the value
is infinite exactly when the play has infinite length, and this property is
independent of time. Hence, costs and rates are of no
importance. $v_k(x'')$ is, thus, correctly set to $\infty$ if $v_k(x)
= \infty$. 
For the remainder of the proof we focus on the case where $v_k(x) <
\infty$. It follows immediately from Lemma~\ref{lemma:connection} that
the value function has the correct form in
the interval $[x',x)$. I.e., $v_k^{\pi'}(x'') = v_k(x) + b^x(k) (x-x'')$,
for $x'' \in [x',x)$.

We next show that $\pi'$ is optimal from time $x'$. From
Lemma~\ref{lemma:nash} we know that it suffices to show that  
$\pi' = (\pi'_1,\pi'_2)$ is a Nash equilibrium from time $x'$. 
Assume for the sake of contradiction that there exists a strategy $\pi_1''$, a
state $k_0$, and a time $x_0 \in [x',x)$, such that $v_{k_0}^{\pi_1'',\pi_2'}(x_0)
< v_{k_0}^{\pi_1',\pi_2'}(x_0)$. Consider the finite play 
\[
\rho_{k_0,x_0}^{\pi_1'',\pi_2'} = (k_0,x_0) \xrightarrow{j_0, \delta_0} (k_1,x_1)
\xrightarrow{j_1, \delta_1} \dots \xrightarrow{j_{t-1}, \delta_{t-1}}
(k_t,x_t)~,
\] 
and assume for
simplicity that, if $x_t \ge x$, a configuration appears at time $x$. Let $\ell > 0$ be
the minimum index such that
$v_{k_{\ell'}}^{\pi_1'',\pi_2'}(x_{\ell'}) \ge v_{k_{\ell'}}^{\pi_1',\pi_2'}(x_{\ell'})$ for all
$\ell' \ge \ell$.
Note that since the play is finite, meaning that $k_t$ is the terminal
state, we have
$v_{k_{t}}^{\pi_1'',\pi_2'}(x_{t})=v_{k_{t}}^{\pi_1',\pi_2'}(x_{t}) = 0$. Hence, $\ell$ is well-defined. Also, since
$(\pi_1',\pi_2')$ is optimal from time $x$, and $x$ appears in a
configuration of the play if $x_t \ge x$, we must have $x_\ell \le x$.

Observe that:
\begin{align*}
v_{k_{\ell-1}}^{\pi_1',\pi_2'}(x_{\ell-1}) &~>~
v_{k_{\ell-1}}^{\pi_1'',\pi_2'}(x_{\ell-1})\\
&~=~ (x_{\ell}-x_{\ell-1}) r_{k_{\ell-1}} + c_{j_{\ell-1}} +
v_{k_{\ell}}^{\pi_1'',\pi_2'}(x_\ell)\\
&~\ge~ (x_{\ell}-x_{\ell-1}) r_{k_{\ell-1}} + c_{j_{\ell-1}} +
v_{k_{\ell}}^{\pi_1',\pi_2'}(x_\ell) ~.
\end{align*}
Continuing this line of thought we next prove that $c_{j_{\ell-1}} +
v_{k_{\ell}}^{\pi_1',\pi_2'}(x_\ell) \ge
v_{k_{\ell-1}}^{\pi_1',\pi_2'}(x_\ell)$, which shows that:
\begin{equation}
v_{k_{\ell-1}}^{\pi_1',\pi_2'}(x_{\ell-1}) ~>~ (x_{\ell}-x_{\ell-1}) r_{k_{\ell-1}} +
v_{k_{\ell-1}}^{\pi_1',\pi_2'}(x_\ell)~.\label{ineq:small}
\end{equation}
We consider two cases: $(i)$ $k_{\ell-1}
\in S_2$, and $(ii)$ $k_{\ell-1} \in S_1$. $(i)$ if $k_{\ell-1}
\in S_2$ then $v_{k_{\ell-1}}^{\pi_1',\pi_2'}(x_\ell) = c_{j_{\ell-1}} +
v_{k_{\ell}}^{\pi_1',\pi_2'}(x_\ell)$ because
$\rho_{k_0,x_0}^{\pi_1'',\pi_2'}$ is also defined by $\pi_2'$. 
The second case $(ii)$ where $k_{\ell-1}
\in S_1$ is slightly more involved.
Assume for the sake of contradiction that
$c_{j_{\ell-1}} + v_{k_{\ell}}^{\pi_1',\pi_2'}(x_\ell) <
v_{k_{\ell-1}}^{\pi_1',\pi_2'}(x_\ell)$. Then $j_{\ell-1}$
is a strongly improving switch w.r.t. $\sigma^x$ for $G^{x,0}$ for
Player 1, and
therefore also for $G^{x}$. This, however, contradicts that neither
player has an improving switch w.r.t. $\sigma^x$ for $G^{x}$.

As argued above we have that $v_k^{\pi'}(x'') = v_k(x) + b^x(k) (x-x'')$,
for all states $k$ and $x'' \in [x',x)$. From (\ref{ineq:small}) we
  therefore get:
\[
v_{k_{\ell-1}}^{\pi_1',\pi_2'}(x_{\ell-1}) ~>~ (x_{\ell}-x_{\ell-1}) r_{k_{\ell-1}} +
v_{k_{\ell-1}}^{\pi_1',\pi_2'}(x_\ell) \quad \Rightarrow
\]
\[
v_{k_{\ell-1}}(x) + b^{x}(k_{\ell-1}) (x-x_{\ell-1}) ~>~ 
(x_{\ell}-x_{\ell-1}) r_{k_{\ell-1}} +
v_{k_{\ell-1}}(x) + b^{x}(k_{\ell-1}) (x-x_{\ell}) \quad
\Rightarrow
\]
\begin{equation}\label{ineq:contradiction}
v_{k_{\ell-1}}(x) + b^{x}(k_{\ell-1}) (x_{\ell}-x_{\ell-1}) ~>~ 
(x_{\ell}-x_{\ell-1}) r_{k_{\ell-1}} +
v_{k_{\ell-1}}(x)
\end{equation}
Let $y = x_{\ell}-x_{\ell-1}$.
From Lemma~\ref{lemma:zero} we know that 
$v_k(x) = a^x(k) = u(k,\sigma^x,G^{x,0})$ for all $k$. It then follows from 
(\ref{eq:G^{x,y}}) that the left-hand-side of
(\ref{ineq:contradiction}) is equal to
$u(k_{\ell-1},\sigma^x,G^{x,y})$. Observe also that the
right-hand-side of (\ref{ineq:contradiction}) is equal to 
$c_{\lambda_{k_{\ell-1}}}^{x,y}$. Hence, (\ref{ineq:contradiction})
shows that $\lambda_{k_{\ell-1}}$ is a strongly improving switch
w.r.t. $\sigma^x$ in $G^{x,y}$. 
Since $x_{\ell-1},x_{\ell} \in [x',x]$ and $x_{\ell-1} < x_{\ell}$
we have $y \in (0,x-x']$. It then follows from
Lemma~\ref{lemma:interval} that neither player has a strongly
improving switch w.r.t. $\sigma^x$ for $G^{x,y}$, and we get a
contradiction. Thus, $\pi'$ is a Nash equilibrium from time
$x'$.

The case for Player 2 is analogous.

Let us note that the same proof would not immediately work for the
more general strategies described in Remark~\ref{rem:common positional strategies}.
\end{proof}

\smallskip\noindent{\bf Proof of correctness of the algorithm.} Lemma~\ref{lemma:induction} allows us to compute optimal strategies by
backward induction once the values $v_k(1)$ at time 1 are
known for all states $k \in S$. Finding $v_k(1)$ and corresponding
optimal strategies from time 1 is, fortunately, not difficult. Indeed, when
$x = 1$ time does not increase further, and we simply solve the priced
game $G'$ that defines $G$. The resulting algorithm is shown in
Figure~\ref{fig:solve}. Note that the choice of first using the
$\ExtendedDijkstra$ algorithm  of Khachiyan {\em et al.}
\cite{KBBEGRZ08} and then the $\StrategyIteration$
algorithm is to facilitate the analysis in Section~\ref{sec:bound}. In
fact, any algorithm for solving priced games could be used.
By observing that $\SolveSPTG$ repeatedly applies
Lemma~\ref{lemma:induction} to construct optimal strategies by
backward induction we get the following theorem.

\begin{theorem}\label{thm:correct}
If $\SolveSPTG$ terminates, it correctly computes the value
function and optimal strategies for both players.
\end{theorem}

Note that $\SolveSPTG$ resembles the sweep-line algorithm of
Shamos and Hoey~\cite{SH76} for the
line segment intersection problem. At every time
$x$ we have $n$ ordered sets of line segments with an intersection
within one set at the next event point $x' =
\NextEventPoint(G^x)$. When handling the event point, the order of the
line segments is updated, and we move on to the next event point.

\begin{figure}
\begin{center}
\parbox{4.7in}{
\begin{function}[H]

$(v(1),(\pi_1(1),\pi_2(1))) \gets \ExtendedDijkstra(G')$\;
$x \gets 1$\;
\While{$x > 0$}{
$(a^x(k) + \epsilon b^x(k),(\sigma_1,\sigma_2)) \gets
  \StrategyIteration(G^{x},(\pi_1(x),\pi_2(x)))$\;
$x' \gets \NextEventPoint(G^x)$\;
\ForAll{$k \in S$ and $x'' \in [x',x)$}{
$v_k(x'') \gets v_k(x) + b^x(k) (x-x'')$\;
$\pi_1(k,x'') \gets \sigma_1(k)$\;
$\pi_2(k,x'') \gets \sigma_2(k)$\;
}
$x \gets x'$\;
}

\Return{$(v,(\pi_1,\pi_2))$}\;
\caption{\SolveSPTG(\mbox{$G$})}
\end{function}
}
\end{center}
\caption{Algorithm for solving a simple priced timed game $G = (G',(r_k)_{k \in S})$.}
\label{fig:solve}
\end{figure}

\subsection{Bounding the number of event points}\label{sec:bound}

Let $G$ be an SPTG. Recall that the only difference between $G^x$ and
$G^{x'}$, for $x \ne x'$, are the costs of actions $\lambda_k$, for
$k\in S$, if $v_k(x) \ne v_k(x')$. The actions
available from each state are therefore the same, and a strategy profile
$\sigma$ for $G^x$ can, thus, also be interpreted as a strategy
profile for $G^{x'}$. To bound the number of event points we assign a
potential to each strategy profile $\sigma$, such that the potential
strictly decreases when one of the players performs a single improving
switch. Furthermore, the potential is defined independently of the
values $v_k(x)$. It then follows that the number of single improving
switches performed by the $\SolveSPTG$ algorithm is at most the total
number of strategy profiles for $G^x$. We further improve this bound
to show that the number of event points is at most exponential in the
number of states. This improves the previous bound by
Rutkowski~\cite{Rutkowski2011}.

\smallskip\noindent{\bf The function $\countt(\sigma,i,\ell,r)$.} Let $n$ be the number of states of $G$, let $N$ be the number of
distinct rates, including rate 0 for the terminal state $\perp$.
Assume that the distinct rates are
ordered such that $r_1 < r_2 < \dots < r_N$. 
Recall that $r(k,\sigma)$ is the rate of the waiting state reached from
$k$ in $\sigma$. Let
\[
\countt(\sigma,i,\ell,r) = |\{k \in S_i \mid \ell(k,\sigma) = \ell ~\wedge~
r(k,\sigma) = r\}|
\]
be the number of states controlled by Player $i$
at distance $\ell$ from $\perp$ in $\sigma$ that reach a waiting state
with rate $r$.

\smallskip\noindent{\bf The potential matrix.} For every strategy profile $\sigma$ for the
priced games $G^x$, for $x \in (0,1]$, define the potential $P(\sigma)
\in \N^{n \times N}$ as an integer matrix as follows. 
\[
P(\sigma)_{\ell, r} = \countt(\sigma,2,\ell,r) - \countt(\sigma,1,\ell,r)
\]
I.e., rows correspond to lengths, columns correspond to rates, and
entries count the number of corresponding Player 2 controlled states
minus the number of corresponding Player 1 controlled states. 

\smallskip\noindent{\bf Ordering the potential matrices.} We define a lexicographic ordering of potential matrices where,
firstly, entries corresponding to lower rates are of higher
importance. Secondly, entries corresponding to shorter lengths are
more important. Formally, we write $P(\sigma) \prec P(\sigma')$ if and
only if there exists $\ell$ and $r$ such that:
\begin{itemize}
\item
$P(\sigma)_{\ell',r'} = P(\sigma')_{\ell',r'}$ for all $r' < r$ and $1
  \le \ell' \le n$.
\item
$P(\sigma)_{\ell',r} = P(\sigma')_{\ell',r}$ for all $\ell' < \ell$.
\item
$P(\sigma)_{\ell,r} < P(\sigma')_{\ell,r}$.
\end{itemize}

\begin{figure}
\[
P(\sigma^{(1)}) = \begin{bmatrix}
  \vspace{0.1cm}
  0 & 0 & 0 & 0\\
  \vspace{0.1cm}
  -1 & 0 & 0 & 0\\
  \vspace{0.1cm}
  -1 & 0 & 0 & 0\\
  \vspace{0.1cm}
  1 & 0 & 0 & 0\\
  0 & 0 & 0 & 0
\end{bmatrix}
\quad
P(\sigma^{(2)}) = \begin{bmatrix}
  \vspace{0.1cm}
  -1 & 0 & -1 & 1\\
  \vspace{0.1cm}
  0 & 0 & -1 & 0\\
  \vspace{0.1cm}
  0 & 0 & 1 & 0\\
  \vspace{0.1cm}
  0 & 0 & 0 & 0\\
  0 & 0 & 0 & 0
\end{bmatrix}
\]

\smallskip

\[
P(\sigma^{(3)}) = \begin{bmatrix}
  \vspace{0.1cm}
  -1 & 1 & 0 & 1\\
  \vspace{0.1cm}
  -1 & -1 & 0 & 0\\
  \vspace{0.1cm}
  0 & 0 & 0 & 0\\
  \vspace{0.1cm}
  0 & 0 & 0 & 0\\
  0 & 0 & 0 & 0
\end{bmatrix}
\quad
P(\sigma^{(4)}) = \begin{bmatrix}
  \vspace{0.1cm}
  -1 & 0 & -1 & 1\\
  \vspace{0.1cm}
  -1 & 0 & 0 & 1\\
  \vspace{0.1cm}
  0 & 0 & 0 & 0\\
  \vspace{0.1cm}
  0 & 0 & 0 & 0\\
  0 & 0 & 0 & 0
\end{bmatrix}
\]
\caption{Example of potential matrices of the strategy profiles from
  Figure~\ref{fig:examplesptg}.
  }
\label{fig:potentials}
\end{figure}

\smallskip\noindent{\bf Example.} Figure~\ref{fig:potentials} shows an example of the potential matrices
of the strategy profiles shown in Figure~\ref{fig:examplesptg}. We use
the following notation:
\begin{itemize}
\item
$\sigma^{(1)}$ is the strategy profile used at 
time $x=1$,
\item
$\sigma^{(2)}$ is the strategy profile used at
time $x\in [2/3,1)$,
\item
$\sigma^{(3)}$ is the strategy profile used at
time $x \in [1/3,2/3)$,
\item
and $\sigma^{(4)}$ is the strategy profile used at
time $x \in [0,1/3)$. 
\end{itemize}
$\sigma^{(1)}$ is also shown in Figure~\ref{fig:examplepricedgame}.
Observe that $P(\sigma^{(1)})_{1,1} = 0$ because states 1 and 5 are
controlled by Player 2 and 1, respectively, and both
move directly to $\perp$, which has rate 0 (thus the entry is 0 because we add $1$ because of state $1$ and subtract $1$ because of state 5).
Also note that the potentials do indeed decrease for the four matrices. At
each event point the strategies are updated for multiple states, however.

\begin{lemma}\label{lemma:potential}
Let $\sigma$ be a strategy profile that is optimal for $G^{x,0}$, for
some $x \in (0,1]$. Let $j \in A^i$ be an improving switch for Player
$i$ w.r.t. $\sigma$ in the priced game $G^x$. Then $P(\sigma[j]) \prec
P(\sigma)$.
\end{lemma}

\begin{proof}
Consider the game $G^x$. Recall that for every strategy profile
$\sigma'$ and state $k \in S$, we let $u(k,\sigma',G^x) = u(k,\sigma',G^{x,0}) + \epsilon
r(k,\sigma')$, where $\epsilon$ is an infinitesimal. Since $\sigma$ is optimal
for $G^{x,0}$ we must have $u(k,\sigma,G^{x,0}) = u(k,\sigma[j],G^{x,0})$ for all $k \in
S$. Indeed, otherwise $j$ would be a strongly improving switch
w.r.t. $\sigma$ in $G^{x,0}$, implying that $\sigma$ is not optimal
for $G^{x,0}$. 

Let $k$ be the state from which the action $j$ originates. It then
follows that $u(k,\sigma) \ne \infty$ and $u(k,\sigma[j]) \ne \infty$.
I.e., it is not possible for exactly one of the payoffs to be infinite,
and if both payoffs are infinite then $j$ would not be an improving
switch.

Assume that $i = 1$. Since $j \in A_k$ is an improving switch for Player 1 we
have $\nu(k,\sigma[j]) < \nu(k,\sigma)$. It is, thus, either the case
that $r(k,\sigma[j]) < r(k,\sigma)$, or that $r(k,\sigma[j]) =
r(k,\sigma)$ and $\ell(k,\sigma[j]) <
\ell(k,\sigma)$. In both cases the most significant
entry $\ell,r$ for which $P(\sigma)_{\ell,r} \ne
P(\sigma[j])_{\ell,r}$ is $\ell = \ell(k,\sigma[j])$ and $r =
r(k,\sigma[j])$. Indeed, all states with new valuations in $\sigma[j]$
move through state $k$ and, thus, have same rates but larger lengths.
Since $i = 1$ we have $P(\sigma)_{\ell,r} < P(\sigma[j])_{\ell,r}$
and, thus, $P(\sigma[j]) \prec P(\sigma)$.

The case for $i = 2$ is similar. $j \in A_k$ is an improving switch
for Player 2, implying that either $r(k,\sigma[j]) > r(k,\sigma)$, or
$r(k,\sigma[j]) = r(k,\sigma)$ and
$\ell(k,\sigma[j]) > \ell(k,\sigma)$. The most significant
entry $\ell,r$ for which $P(\sigma)_{\ell,r} \ne
P(\sigma[j])_{\ell,r}$ is then $\ell = \ell(k,\sigma)$ and $r =
r(k,\sigma)$. Since $i = 2$ we again have $P(\sigma)_{\ell,r} <
P(\sigma[j])_{\ell,r}$ and subsequently $P(\sigma[j]) \prec
P(\sigma)$.

\end{proof}

\begin{theorem}\label{thm:exp}
The total number of event points for any SPTG $G$ with $n$ states is
$\EventPoints(G) \le \min \{12^n, \prod_{k \in S}
(|A_k|+1)\}$. Furthermore, if there is only one player then $\EventPoints(G) = O(n^2)$.
\end{theorem}
\begin{proof}
Consider the variant of the $\SolveSPTG$ algorithm where
$\StrategyIteration$ only performs single
improving switches for both players. I.e., when solving $G^x$, for
some $x \in (0,1]$, Player 1 performs one improving switch, then
Player 2 repeatedly performs single improving switches as long as
possible, and then the process is repeated. The resulting optimal
strategy profile $\sigma^x$ is then used as the starting point for
solving the next priced game $G^{x'}$, for $x' =
\NextEventPoint(G^x)$.

Once the initial strategy profile $\sigma = (\pi_1(1),\pi_2(1))$ is
found, any strategy profile $\sigma'$ that is subsequently produced by
the $\StrategyIteration$ algorithm at some time $x$ is optimal for the
priced game $G^{x,0}$. I.e., $\sigma^x$ is optimal for all $G^{x''}$
with $x'' \in (x',x]$, where $x' = \NextEventPoint(G^x)$. In
particular, the payoffs resulting from $\sigma^x$ and $\sigma^{x'}$ in
$G^{x'}$ only differ by some second order term. Hence, we can apply
Lemma~\ref{lemma:potential} to the strategy profiles, and conclude
that the potential decreases with every improving switch. From this we
immediately get that the total number of strategy profiles in $G^x$, $\prod_{k \in S} (|A_k|+1)$, is an upper
bound on $\EventPoints(G)$.

We next show that $\EventPoints(G) \le 12^n$. A matrix $P \in \N^{n
  \times N}$ corresponding to
a legal potential can always be constructed in the following
way. Let each entry $(\ell,r)$ be associated with a set $S_{\ell,r}$
of corresponding states. I.e., $S_{\ell,r}$ contains the states for
which it takes $\ell$ moves to reach $\perp$ in the priced game, and
the rate encountered is the $r$'th smallest rate of the game.
Pick a non-empty subset of the columns $C \subseteq
\{1,\dots,N\}$. This will be the columns, such that in column $r$, there is an $\ell$ such that $\countt(\sigma,2,\ell,r)\neq 0$ or $\countt(\sigma,1,\ell,r)\neq 0$. This can be done in at most $2^N-1 \le 2^{n+1}$ ways. Next,
assign states to the sets of the entries. If $S_{\ell,r} \ne
\emptyset$, then we must also have $S_{\ell',r} \ne \emptyset$ for all
$\ell' < \ell$, by definition. This allows us to assign states to sets in an ordered
way. Let $(\ell,r)$ be the current entry starting
from $\ell = 1$ and $r = \min ~ C$. The current entry will be lexicographic increasing in $(r,\ell)$. Repeatedly add a state from either
$S_1$ or $S_2$ to
$S_{\ell,r}$ and update the current entry in one of the following three
ways: 
\begin{itemize}
\item Do nothing: More states will be assigned to $S_{\ell,r}$.
\item Move to the next row: No more states will be assigned to
 $S_{\ell,r}$, but some will be assigned to $S_{\ell,r+1}$.
\item Move to the beginning of the next column of $C$:
No more states will be assigned to $S_{\ell,r'}$ for any $r'$.
\end{itemize}
There are $n$ (one for each state in the game) such
iterations, and in each iteration there are at most six possible
options. Hence, the states can be added in at most $6^n$ ways.
Furthermore, we do not need to update the current entry after the last
state has been added, which saves us a factor of 3. The
total number of possible matrices $P$ is, thus, at most $12^n$.

When there is only one player $i$ the argument becomes much
simpler. Observe that the rates change monotonically when going back
in time: if $i = 1$ the rates decrease, and if $i = 2$ the rates
increase. Furthermore, at every event point at least one state changes
rate. Hence, there can be at most $nN \le n(n+1)$ event points.

\end{proof}

\begin{theorem}\label{thm:time}
$\SolveSPTG$ solves any SPTG $G$ in time $O(m \cdot \min \{12^n,
\prod_{k \in S} (|A_k|+1)\})$ in the unit cost
model, where $n$ is the number of states and $m$
is the number of actions. Alternatively, the variant of $\SolveSPTG$
that uses the $\ExtendedDijkstra$ algorithm instead of
$\StrategyIteration$ solves $G$ in
time $O(\EventPoints(G) (m + n \log n))$.
\end{theorem}
\begin{proof}
The correctness of $\SolveSPTG$ follows from
Theorems~\ref{thm:correct} and \ref{thm:exp}.

For the first bound we get from the proof of Theorem~\ref{thm:exp}
that, in fact, not only the number of event points, but also the
number of single improving switches is bounded by $\min \{12^n,
\prod_{k \in S} (|A_k|+1)\}$. Valuations for a strategy profile $\sigma$
can be computed in time $O(n)$, and then the next event point can be
computed in time $O(m)$. I.e., for each $k \in S$ we find the next
event point at time $x$ among the intersections of $f_{\sigma(k),x}$
and $f_{j,x}$, for $j \in A_k$.

For the second bound we are using the $\ExtendedDijkstra$ algorithm of 
Khachiyan {\em et al.}~\cite{KBBEGRZ08} instead of $\StrategyIteration$
in the inner while-loop. The $\ExtendedDijkstra$
algorithm has the same complexity as Dijkstra's algorithm\footnote{To
  get this bound for the Extended 
Dijkstra's algorithm, 
actions of the maximizer should not be inserted into the priortiy queue. Instead,
a choice of action for the maximizer for a state is fixed when the values of all
possible successors of that state are known.}.
Fredman and Tarjan~\cite{FT87} showed that, using
Fibonacci heaps, Dijkstra's algorithm can solve the single source
shortest path problem for a graph with $n$ vertices and $m$ edges in
time $O(m + n \log n)$.

\end{proof}

Theorem~\ref{thm:opt} follows as a corollary of
Theorem~\ref{thm:time}, since $\SolveSPTG$ is always guaranteed to
compute optimal strategies, and the resulting value functions
are continuous piecewise linear functions.

\section{Priced timed games\label{ptg}}
\emph{One-clock priced timed games} (PTGs) extend SPTGs in two ways. First, actions are associated with time intervals
during which they are available, and second, certain actions will
cause the time to be reset to zero. Also, we do not require the time
to run from zero to one.

\smallskip\noindent{\bf Formal definition of PTGs.} Formally, a PTG $G$ can be described by a
tuple \[G=(S_1,S_2,(A_k)_{k \in S}, (c_j)_{j\in A}, d,(r_k)_{k\in
  S},(I_j)_{j\in A},R),\] where $S = S_1 \cup S_2$ and $A=\bigcup_{k \in S}
A_k$. The complete description 
of the individual components of $G$ is as follows. Note that only the
last two components are new compared to priced games and SPTGs.
\begin{itemize}
\item
$S_i$, for $i \in \{1,2\}$, is a set of states controlled by Player $i$.
\item
$A_k$, for $k \in S$, is a set of actions available from
  state $k$.
\item
$c_j \in \R_{\ge 0} \cup \{\infty\}$, for $j \in A$, is the cost of action $j \in A$.
\item
$d: A \to S \cup \{\perp\}$ is a mapping from actions to destinations
  with $\perp$ being the terminal state.
\item
$r_k \in \R_{\ge 0}$, for $k \in S$, is the rate for waiting at state $k$.
\item 
$I_j$, for $j \in A$, is the \emph{existence interval} (a real interval)
of action $j$ during which it is available.

\item
$R \subseteq A$ is the set of \emph{reset actions}.
\end{itemize}

\smallskip\noindent{\bf The class $(n,m,r,d)$-PTG.} To simplify the statements of many of the remaining lemmas we let $(n,m,r,d)$-PTG be the class of all PTGs consisting of $n$ states, $r$ of which are the destination of some reset action, $m$ actions and $d$ distinct endpoints of existence intervals.

\smallskip\noindent{\bf End points of intervals.} We let $e(I)$ be the set of endpoints of interval $I$, and define
$M=\max_{\cup_{j\in A}} e(I_j)$. I.e., after time $M$ no actions are
available and the game must end.
Note that PTGs are often defined with existence intervals for both
states and actions. For convenience, we decided to omit this feature
since it is not difficult to translate between the two version.

\smallskip\noindent{\bf Playing PTGs.} PTGs are played like SPTGs with the exception that using a reset
transition resets the time to zero and that the actions must be available when used. We, thus, again operate with
configurations $(k, x) \in S \times [0,M]$ corresponding to a pebble
being placed on state $k$ at time $x$. The player controlling state
$k$ chooses an action $j \in A_k$ and a delay $\delta \ge 0$, such
that $j$ is available at time $x+\delta$. I.e., $x+\delta \in I_j$.
We assume for simplicity that such an action is always available. The
pebble is then moved to state $d(j)$, the time is incremented to
$x+\delta$ if $j \not\in R$ and reset to zero otherwise, and the play
continues. The game ends when the terminal state $\perp$ is
reached.

\smallskip\noindent{\bf Play and outcomes.} We again let a play be a sequence of legal steps starting from some
configuration $(k_0,x_0)$:
\[
\rho = (k_0,x_0) \xrightarrow{j_0, \delta_0} (k_1,x_1)
\xrightarrow{j_1, \delta_1} \dots 
\]
where, for all $\ell \ge 0$, $x_\ell + \delta_\ell \in I_{j_\ell}$,
and if $j_{\ell} \in R$ then $x_{\ell+1} = 0$. The outcome of infinite
plays and finite plays ending at the terminal state $\perp$ are defined
analogously to SPTGs.

\smallskip\noindent{\bf Positional strategies.} Let $\Plays(i)$ be the set of finite plays ending at a state controlled by
Player $i$. Note that $\rho \in \Plays(i)$ specifies the current state
and time, as well as the history leading to this configuration. A (positional) strategy for Player $i$ is again defined as
a map $\pi_i: S_i \times [0,M] \to A \cup \{\lambda\}$ from
configurations of the game to choices of actions. Again, for every $k \in S_i$ and $x \in [0,1)$, if $\pi_i(k,x) = \lambda$
then we require that there exists a $\delta > 0$ such that for all $0
\le \epsilon < \delta$, $\pi_i(k,x+\epsilon) = \lambda$.
Let $\delta_{\pi_i}(k,x) = \inf \{x'-x \mid x \le x' \le 1, \pi_i(k,x') \ne
\lambda\}$ be the delay before the pebble is moved when starting in
state $k$ at time $x$ for some strategy $\pi_i$.  Previous works have defined such strategies in other ways, see Remark \ref{rem:common positional strategies}.

\smallskip\noindent{\bf History-dependent strategies.} A
\emph{history-dependent strategy} for Player $i$ is a map $\tau_i:
\Plays(i) \to (A,\R_{\ge 0})$ that maps every play $\rho$ ending in a
state $k \in S_i$ to an action $j \in A_k$ and a delay $t$. We will only use history-dependent strategies in the proof of one lemma (Lemma \ref{lem:remove resets}). Note that
history-dependent strategies generalize positional strategies. We denote the set of history-dependent
strategies for Player $i$ by $T_i(G)$, where $G$ is omitted if it is
clear from the context. Similarly, the set of positional
strategies for Player $i$ is denoted by $\Pi_i(G)$.

\smallskip\noindent{\bf Values.} Let $\rho_{k,x}^{\tau_1,\tau_2}$ be the play generated when, starting
from $(k,x)$, the players play according to $\tau_1$ and $\tau_2$. The
corresponding value function is again defined as:
\[
v_{k}^{\tau_1,\tau_2}(x) = \cost(\rho_{k,x}^{\tau_1,\tau_2}).
\]
Best response, lower and upper value functions are again defined as:
\begin{align*}
v_k^{\tau_1}(x) &= \sup_{\tau_2 \in T_2} ~ v_k^{\tau_1, \tau_2}(x) \\
v_k^{\tau_2}(x) &= \inf_{\tau_1 \in T_1} ~ v_k^{\tau_1, \tau_2}(x) \\
\underline{v}_k(x) &= \sup_{\tau_2 \in T_2} ~ v_k^{\tau_2}(x) = \sup_{\tau_2 \in T_2} ~ \inf_{\tau_1 \in T_1}
~ v_k^{\tau_1, \tau_2}(x)\\
\overline{v}_k(x) &= \inf_{\tau_1 \in T_1} ~ v_k^{\tau_1}(x) = \inf_{\tau_1 \in T_1} ~ \sup_{\tau_2 \in T_2} ~
v_k^{\tau_1, \tau_2}(x)
\end{align*}


Bouyer \emph{et al.}~\cite{BLMR06} proved the following fundamental
theorem.

\begin{theorem}[Bouyer \emph{et al.}~\cite{BLMR06}]\label{thm:optimalvalues}
For every PTG $G$, there exist value functions
$v_k(x) := \underline{v}_k(x) = \overline{v}_k(x)$. Moreover, a player
can get arbitrarily close to the values even when restricted to
playing positional strategies:
\begin{align*}
\underline{v}_k(x) &= \sup_{\pi_2 \in \Pi_2} ~ \inf_{\tau_1 \in T_1}
~ v_k^{\tau_1, \pi_2}(x)=\sup_{\pi_2 \in \Pi_2} ~ \inf_{\pi_1 \in \Pi_1}
~ v_k^{\pi_1, \pi_2}(x)\\
\overline{v}_k(x) &= \inf_{\pi_1 \in \Pi_1} ~ \sup_{\tau_2 \in T_2} ~
v_k^{\pi_1, \tau_2}(x)= \inf_{\pi_1 \in \Pi_1} ~ \sup_{\pi_2 \in \Pi_2} ~
v_k^{\pi_1, \pi_2}(x)
\end{align*}
\end{theorem}

For the purpose of solving PTGs it, thus, suffices to consider
positional strategies. In the remainder of this section we will
therefore restrict ourselves to positional strategies unless otherwise
specified.


\smallskip\noindent{\bf $\epsilon$-optimal strategies.} A strategy $\pi_i \in \Pi_i$ is \emph{$\epsilon$-optimal} for
Player $i$ for $\epsilon\geq 0$ if:
\[
\forall k \in S,~ x \in [0,M]: \quad
|v_k^{\pi_i}(x) - v_k(x)|\leq\epsilon.
\]
Since PTGs have value functions, $\epsilon$-optimal
strategies always exist for both players, for any $\epsilon > 0$.

\smallskip\noindent{\bf Non-existence of optimal strategies.} Optimal strategies do not always exist, as shown 
by Bouyer {\em et al.}~\cite{BLMR06}. Indeed,
consider the PTG shown in Figure~\ref{fig:nonexistence}. State 1 is
controlled by Player 1, the minimizer, and state 2 is controlled by
Player 2, the maximizer. The value functions are shown on the
right. Two actions leading to the terminal state are available from
state 2 at time 0 and 1, respectively. Since the rate of state 2 is
0, Player 2 picks the
more expensive action with cost $c_2 = 1$ at time 0, and at other
times Player 2 waits until time 1 and picks the cheaper action with
cost $c_3 = 0$.
From state 1 exactly one action is available at all times, and since
the rate is 1, Player 1 leaves the state \emph{as soon as possible},
only not at time 0. Since no strategy can implement leaving as soon as
possible there is no optimal strategy for Player 1. More precisely,
for every waiting time $\delta$ chosen by Player 1 at time 0, there
exists a smaller waiting time $\delta' < \delta$ that achieves a
better value.

\begin{figure}
\center
\includegraphics{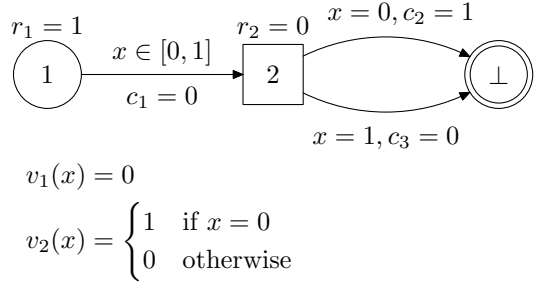}
\caption{Example of a PTG with no optimal strategy profile.}
\label{fig:nonexistence}
\end{figure}

We reduce solving any PTG to solving a
number of SPTGs. The first step towards this goal is to remove reset
actions by extending the game.
\begin{lemma}
\label{lem:remove resets}
Let $G$ be a $(n,m,r,d)$-PTG. Solving $G$ can be reduced to solving $r+1$ $(n,m,0,d)$-PTGs.
\end{lemma}

\begin{proof}
Let $\pi = (\pi_1,\pi_2)$ be any strategy profile, and suppose the play
$\rho^\pi_{k_0,x_0}$ is using two reset actions $j,j' \in R$ leading
to the same state $d(j) = d(j') = k$. Then the configuration $(k,0)$
appears twice in $\rho^\pi_{k_0,x_0}$, and since strategies are
history-independent it appears an infinite number of times. It follows
that $v_{k_0}^\pi(x_0) = \infty$.
By the pigeonhole principle we get that if a play $\rho^\pi_{k_0,x_0}$
uses $r+1$ reset actions, then some state is visited twice by some reset
actions, and therefore $v_{k_0}^\pi(x_0) = \infty$.

Thus, when playing $G$ we may augment configurations by the number of
times a reset action has been used, and once this number reaches $r+1$
we may assume without loss of generality that the value is
infinite. This defines a new PTG $G'$ with states $S' = S \times
\{0,\dots,r\}$ and actions $A' = A \times \{0,\dots,r\}$ in the
following natural way. For $j \in A$ and $\ell \in \{0,\dots,r\}$,
destinations and costs are defined as
\begin{align*}
d'(j,\ell) &= \begin{cases}
(d(j),\ell+1) & \text{if $j \in R$ and $\ell < r$}\\
\perp & \text{if $j \in R$ and $\ell = r$}\\
(d(j),\ell) & \text{otherwise}
\end{cases}\\
c'_{(j,\ell)} &= \begin{cases}
\infty & \text{if $j \in R$ and $\ell = r$}\\
c_j & \text{otherwise}
\end{cases}
\end{align*}
while rates, existence intervals and reset actions are the same as for
the corresponding states and actions of $G$. Plays and value
functions of $G'$ will be denoted by $\rho'$ and $v'$,
respectively. We will show that for all $(k,x) \in S \times [0,M]$,
$v'_{(k,0)}(x) = v_k(x)$.

Every strategy profile $\pi'$ for $G'$ can be interpreted as a
history-dependent strategy profile for $G$ in the following
way: For every play that can be achieved by moving according to $\pi'$
make the corresponding choice in $\pi'$, for other plays make arbitrary
choices. 
Also, every positional strategy profile $\pi$ for $G$ can
be interpreted as a strategy profile for $G'$ by using the same
choices regardless of the number of encountered reset actions.
With these interpretations we see that
$\Pi_i(G) \subseteq \Pi_i(G') \subseteq T_i(G)$. 

For all configurations $(k,x) \in S \times [0,M]$, if
${\rho'}_{(k,0),x}^{\pi'}$ uses at most $r$ reset actions, then
$\cost({\rho'}_{(k,0),x}^{\pi'}) = \cost({\rho}_{k,x}^{\pi'})$, since
the actions encountered in the two games have the same costs. If 
${\rho'}_{(k,0),x}^{\pi'}$ uses more than $r$ reset actions, then
$\cost({\rho'}_{(k,0),x}^{\pi'}) = \infty \ge
\cost({\rho}_{k,x}^{\pi'})$. Hence, we always have 
${v'}_{(k,0)}^{\pi'}(x) \ge v_{k}^{\pi'}(x)$.
Using Theorem~\ref{thm:optimalvalues} it
follows that:
\begin{align*}
v'_{(k,0)}(x) &= \inf_{\pi'_1 \in \Pi'_1} ~ \sup_{\pi'_2 \in \Pi'_2} ~
{v'}_{(k,0)}^{\pi'_1, \pi'_2}(x) \ge
\inf_{\pi'_1 \in \Pi'_1} ~ \sup_{\pi'_2 \in \Pi'_2} ~
v_{k}^{\pi'_1, \pi'_2}(x) \\ &=
\inf_{\pi_1 \in \Pi_1} ~ \sup_{\pi_2 \in \Pi_2} ~
v_{k}^{\pi_1, \pi_2}(x) = v_{k}(x)
\end{align*}
The first inequality follows from the costs being larger in $G'$, and
the next equality follows Theorem~\ref{thm:optimalvalues}; the same
values can be obtained using only positional strategies in $G$.

Next we show that $v'_{(k,0)}(x) \le v_{k}(x)$, implying that
$v'_{(k,0)}(x) = v_{k}(x)$. This is clearly true if $v_{k}(x) =
\infty$, thus, we may assume that $v_{k}(x) < \infty$. In particular,
$\epsilon$-optimal strategies do not generate plays with more than $r$
reset actions in neither $G$ nor $G'$. We see that:
\begin{align*}
v_{k}(x) &= \inf_{\pi_1 \in \Pi_1} ~ \sup_{\pi_2 \in \Pi_2} ~
v_{k}^{\pi_1, \pi_2}(x) 
= 
\inf_{\pi'_1 \in \Pi'_1} ~ \sup_{\pi'_2 \in \Pi'_2} ~
v_{k}^{\pi'_1, \pi'_2}(x) \\
&=
\inf_{\pi'_1 \in \Pi'_1} ~ \sup_{\pi'_2 \in \Pi'_2} ~
{v'}_{(k,0)}^{\pi'_1, \pi'_2}(x) 
=v'_{(k,0)}(x)
\end{align*}
For the second equality we use Theorem~\ref{thm:optimalvalues}; the
values do not change even if certain history-dependent strategies are
available. For the third equality we use the assumption that the
values are finite. This implies that for relevant strategy profiles
the values of the two games are the same.

We now know that in order to find the value $v_k(x)$ in $G$ it
suffices to find $v'_{(k,0)}(x)$ in $G'$. To do this we exploit the
special structure of $G'$. We observe that states $(k,\ell) \in S
\times \{0,\dots,r\}$ do not depend on states $(k,\ell')$ with $\ell'
< \ell$. Thus, the game can be solved using backward induction on
$\ell$. In particular, when $v'_{(k,\ell+1)}(x)$ is known for all $k$
and $x$, then the subgame consisting of states $(k,\ell)$, for $k \in
S$, can be viewed as an independent PTG with no reset actions. I.e.,
reset actions lead to states with known values, and can, thus, be
thought of as going directly to the terminal state with an appropriate
cost. Each subgame has $n$ states and $m$ actions, and there are $r+1$
such subgames.

\end{proof}

\smallskip\noindent{\bf Overview over remaining lemmas.} 
Now we just need to show how to solve PTGs without resets using SPTGs. 
We will show the statement using 3 reductions. First we will reduce
PTGs without resets to the subclass of such games, where, for each
action $j\in A$, we have $I_j\in \{(0,1),[1,1]\}$. Afterwards we will
reduce further to the subclass of PTGs where for each action $j\in A$,
we have $I_j\in \{[0,1],[1,1]\}$. At the end we will reduce those to SPTGs.

Let $X$ be the set which consists of $0$ and the endpoints of existence intervals of $G$. 
Let the $i$'th largest element in $X$ be $M_i$. Note that $M_1=M$.

We will now define some functions on PTGs. For a PTG $G=(S_1,S_2,$ $(A_k)_{k \in S},$ $c,$ $ d,$ $(r_k)_{k\in
  S},$ $(I_j)_{j\in A},R),$ where $R=\emptyset$, a $x\in \R$ and a vector $v\in (\R_{\ge 0\cup\{\infty\}})^n$, let the priced game $G^{v,x}=(S_1,S_2,(A_k')_{k\in S},c',d')$ be defined by:

\begin{alignat*}{2}
\forall k\in S&:& A_k'&=\{j\in A_k \mid x\in I_j\}\cup \{\perp_k\}\\
\forall j\in A_k'&:& c_j' &=\begin{cases}
v_k & \text{if $j=\perp_k$}\\
c_j & \text{otherwise}
\end{cases}\\
\forall j\in A_k'&:& d_j' &=\begin{cases}
\perp & \text{if $j=\perp_k$}\\
d_j & \text{otherwise}
\end{cases}
\end{alignat*}

The game $G^{v,x}$ is similar to the priced game defined in Definition \ref{def:Gxe}. The intuition is that $G^{v,x}$ can model a specific moment in time. 

\begin{definition}\label{def:sptgvxd}

For a given PTG $G=(S_1,S_2,(A_k)_{k \in S}, $ $c, d,(r_k)_{k\in
  S},(I_j)_{j\in A},R),$ a $x\in \R$ and a vector $v\in (\R_{\ge
  0\cup\{\infty\}})^n$, let the SPTG
$G^{v,x,d}=(S_1,S_2',$ $(A_k')_{k \in (S_1\cup S_2')}, $ $c', d',(r_k')_{k\in (S_1\cup S_2')})$ be defined by:

\begin{alignat*}{2}
&&S_2'&=S_2\cup \{\max\}\\
\forall k\in S&:& A_k' &=\{j\in A_k \mid x\in I_j\}\cup \{\perp_k\}\\
&&A_{\max}&=\{\perp_{\max} \}\\
\forall j\in A_k'&:& c_j' &=\begin{cases}
0 & \text{if $k=\max$}\\
v_k & \text{if $k\neq \max$ and $j=\perp_k$}\\
c_j & \text{otherwise}\\
\end{cases}\\
\forall j\in A_k'&:& d_j'&=\begin{cases}
\max & \text{if $k\in S_1$ and either $j=\perp_k$} \\ & \text{or $d_j=\perp$}\\
\perp & \text{if $k\in S_2$ and either $j=\perp_k$} \\ & \text{or $d_j=\perp$}\\
\perp & \text{if $k=\max$}\\
d_j & \text{otherwise}
\end{cases}\\
\forall k\in S&:& r_k'&=r_k\cdot d\\
&&r_{\max} ' &=\max_{k\in S}\{r_k\}\\
\end{alignat*}
\end{definition}

The game $G^{v,x,d}$ is constructed from the proof of Lemma \ref{lem:remove intervals},  Lemma \ref{lem:remove intervals 2} and Lemma \ref{lem:remove intervals 3}. The intuition is that the game can model an arbitrary length interval, in the original game, where no action changes status between available and unavailable.

\begin{figure}
\begin{center}
\parbox{4.7in}{
\begin{function}[H]
$v(M_1) \gets \ExtendedDijkstra((S_1,S_2,(\{j\in A_k \mid M_1\in I_j\})_{k\in S},c,d))$\;
$i\gets 1$\;
\While{$M_i> 0$}{
$i\gets i+1$\;
$x\gets \frac{M_{i-1}+M_i}{2}$\;
$v' \gets \ExtendedDijkstra(G^{v(M_{i-1}),x})$\;
$v^*(x)\gets \SolveSPTG(G^{v',x,M_{i-1}-M_i})$\;
\ForAll{$x \in (M_i,M_{i-1})$}{
$v(x)\gets v^*(\frac{x-M_{i}}{M_{i-1}-M_{i}})$\;
}
$v(M_i)\gets \ExtendedDijkstra(G^{v^*(0),M_i})$\;
}
\Return{$v$}\;
\caption{\SolvePTG(\mbox{$G$})}
\end{function}
}
\end{center}
\caption{Algorithm for solving PTGs without reset actions.}
\label{fig:solvePTG}
\end{figure}

\begin{theorem}
The algorithm in Figure \ref{fig:solvePTG} correctly solves Priced Timed Games without reset actions.
\end{theorem}

The proof of correctness is that the algorithm is a formalization of the reductions in Lemma \ref{lem:remove intervals},  Lemma \ref{lem:remove intervals 2} and Lemma \ref{lem:remove intervals 3}.
Note that instead of $\frac{M_i+M_{i-1}}{2}$, any arbitrary point
inside $(M_i,M_{i-1})$ would work. 

Let $PTG_{I}^{n,m}$ be the subclass of PTGs, consisting of $n$ states, $m$ actions, none of which are reset actions, and where the existence interval for each action, $j$, is either $I_j=I$ or $I_j=[1,1]$. {\em In the latter case $d_j=\perp$.} Note that for such games we can WLOG assume that $m\leq 2n^2$, because for all actions with the same existence interval, only the one with the best cost will be used. 

The following lemma is using techniques similar to those of the region abstraction algorithm of Laroussinie, Markey, and Schnoebelen \cite{LMS04}. The lemma is presented here mainly to be more explicit about the class of games reduced to.

\begin{lemma}
\label{lem:remove intervals}
Any game $G$ in $(n,m,0,d)$-PTG can be solved in time $O((n\log n +\min (m,n^2))d)$ using at most $d$ calls to an oracle, $R$, that solves $PTG_{(0,1)}^{n,m+n}$.
\end{lemma}

We sketch the proof. It is easy to find the value of $k \in S$ at time
$M_1$ in a priced timed game without reset actions, because no player
can wait and hence the game is equivalent to a priced game. Between
time $M_2$ and $M_1$ the game is nearly an SPTG,
since we can simply translate by decreasing all times with $M_2$ and
divide the times by $M_1-M_2$ to get a game between 0 and 1
instead. After finding the value between $M_2$ and $M_1$ we can then
find the value at $M_2$, since we know the cost if we wait (it becomes
$\lim_{x\rightarrow M_2^{+}} v(k,x)$), by viewing the game as a priced
game at that point.  We can then find the value between $M_3$ and
$M_2$, then at time $M_3$ and so on, until we have solved the
game.

\begin{proof}
We can find $v(k,M_1)$ as the value of state $k$ in the priced game which consists of the same states as $G$ and the actions available at time $M_1$.  We can do so, because the game contains no reset actions and we can therefore neither increase nor decrease time. Note that if multiple actions, $j$, in $A_k$ and $d_j=\ell$ exists for $k,\ell \in S$, we can ignore all but the one with the best cost for the controller of $k$. Hence we can solve such a priced game in time $O(n\log n+\min (m,n^2))$.

We now want to find $\forall k\in S, x\in (M_2,M_1): v(k,x)$. We see that if we wait until $M_1$ in some state, $k$, the rest of the path to $\perp$ costs $v(k,M_1)$, if we play optimally from $M_1$. We see that if we start at a time $x$, we can not reach a time before $x$, because there are no reset actions. Hence, look at a modified game $G'$, with value function $v'$: $G'$ consists of the same set of states as $G$, but it only has the actions available in the interval $(M_2,M_1)$, which, in $G'$, only exists in that interval, and for each state, $k$, an action to $\perp$ of cost $v(k,M_1)$ which is only available at time $M_1$. We will also modify $G'$ such that we subtract $M_2$ from all points in time. Clearly that will not matter for plays starting after time $M_2$. Note that all intervals for actions are either $(0,M_1-M_2)$ or $[M_1-M_2,M_1-M_2]$. We can also divide all points in time with $M_1-M_2$, by also multiplying the rate of each state with $M_1-M_2$. Hence all existence intervals either have the form $(0,1)$ or $[1,1]$ and we clearly have that 
\[\forall x\in (M_2,M_1),k\in S: v(k,x)=v'\big( k,\frac{x-M_2}{M_1-M_2}\big) .\]

We can solve $G'$ using a call to $R$. 

We will now find $v(k,M_2)$. If it is optimal to wait at time $M_2$ in state $k$, we have that $v(k,M_2)=\lim_{x\rightarrow 0^+}v'(k,x)=v'(k,0)$, because we might as well wait as little as possible and then play optimally from there. Hence, $v(k,M_2)$ is the value of state $k$ in the priced game $G''$, consisting of the same states as $G$ and the same actions as those available at time $M_2$ in $G$ and for each state $k$, a action from $k$ to $\perp$ of cost $v'(k,0)$. Like we did for $M_1$ we can ignore all but one action from a state to another. Hence we can solve such a priced game in time $O(n\log n+\min (m,n^2))$.

We now want to find $\forall k\in S, x\in (M_3,M_2): v(k,x)$. We can therefore do like we did for $\forall k\in S, x\in (M_2,M_1): v(k,x)$. Also to find $v(k,M_3)$ we can do like we did for $v(k,d_2)$. We keep on doing this until we are done.

Hence, we use $d$ calls to $R$ and solve $d+1$ priced games. 
\end{proof}

We will now to reduce a game in $PTG_{(0,1)}^{n,m}$, with value function, $v$, to a game in $PTG_{[0,1]}^{n,m}$ using $O(n\log n+\min (m,n^2))$ time.
First note that it is easy to find $v(k,1)$ using a priced game, because time can not change at time 1. 
It is clear that $v(k,0)=\lim_{x\rightarrow 0^+} v(k,x)$, because the only option at time 0 is to wait.
Hence we only need to look at finding $v(k,x)$ for $x\in (0,1)$.  To that we will use the following lemma. Note that the game $G'$ mentioned in the lemma is in $PTG_{[0,1]}^{n,m}$.

\begin{lemma}
\label{lem:remove intervals 2}
Let $G$ be a $PTG_{(0,1)}^{n,m}$, with value function $v$.
Let $G'$ be the modified version of $G$, where all existence intervals
of the form $(0,1)$ in $G$ instead have the form $[0,1]$. Let $v'$ be
the value function for $G'$. We then have: $\forall k\in S,x\in (0,1):v(k,x)=v'(k,x).$
\end{lemma}

Intuititively, the proof is as follows: First notice that the extension downward of the existence interval does not change anything, since we can not get to a earlier point in time, than we already are at. The upward extension is somewhat more complicated, but for any $\epsilon>0$ a $\frac{\epsilon}{2}$-optimal strategy in $G$ can be modified in a slim interval close to 1 to yield a strategy which is $\epsilon$-optimal strategy in $G'$. We can then from that make an argument that the values are within $\epsilon$ of each other.

\begin{proof}
Let $\epsilon>0$. We will show that $\forall k\in S,x\in (0,1):v(k,x)=v'(k,x)$ by constructing a strategy, $\sigma_1$, for player 1 that guarantees at most $v'(k,x)+\epsilon$, in $G$, for any $k\in S$ and for any $x\in (0,1)$. Similarly we will construct a strategy, $\sigma_2$, for player 2 that guarantees at least $v'(k,x)-\epsilon$, in $G$, for any $k\in S$ and for any $x\in [0,1)$.

Let $\sigma_1'$ be a $\epsilon/2$-optimal strategy for player 1 in $G'$. Let $r_{\max}=\max_{k\in S} r(s)$. Let $\sigma_1^G$ be the optimal strategy in the priced game which consists of the same states as $G$, but only those actions available at time 1. Let $\sigma_1^{G'}$ be the optimal strategy in the priced game which consists of the same states as $G'$. It is clear that if the existence interval of $\sigma_1^{G'}(k)$ in $G$ is $[1,1]$ then $\sigma_1^{G'}(k)=\sigma_1^{G}(k)$. 

We will now construct $\sigma_1$.
\begin{alignat*}{2}
\sigma_1(k,x)&=&\begin{cases}
\lambda & \text{if $x=0$}\\
\sigma_1'(k,x) & \text{if $0<x< 1-\frac{\epsilon}{2 r_{\max}}$}\\
\sigma_1^{G''}(k) & \text{if $1-\frac{\epsilon}{2 r_{\max}}\leq x<1$ and the existence}\\
& \text{interval for $\sigma_1^{G'}(k)$ in $G$ is $(0,1)$}\\
\lambda & \text{if $1-\frac{\epsilon}{2 r_{\max}}\leq x<1$  and the existence}\\
& \text{interval for $\sigma_1^{G'}(k)$ in $G$ is $[1,1]$}\\
\sigma_1^G(k) & \text{if $x=1$}\\
\end{cases}
\end{alignat*}

Let $k\in S, x\in (0,1)$. We will first show that $v'(k,x)\geq v'(k,1)$. If $k\in S_2$ player 2 could simply wait until time 1, and since $r(k)\geq 0$ the statement follows. If $k\in S_1$ player 1 must keep the play away from $S_2$ (because that would reduce it to the first case) and have no advantages in waiting, since no new actions become available. But since all actions are available at time 1, player 1 could follow the same strategy, as he uses at time $x$, and get the same cost.

We will now show that $\sigma_1$ guarantees at most $v'(k,x)+\epsilon$, for $k\in S, x\in(0,1)$ in $G$. We will do so by contradiction. Assume not. Hence there is a strategy $\sigma_2$, a $x\in (0,1)$ and a $k\in S$, such that $v_k^{(\sigma_1,\sigma_2)}(x)>v'(k,x)+\epsilon$. Let $\rho$ be the play defined by $(\sigma_1,\sigma_2)$. 

There are now two cases. Either $x\geq 1-\frac{\epsilon}{2 r_{\max}}$ or not. 

If $x\geq 1-\frac{\epsilon}{2 r_{\max}}$, we know that \[v_k^{(\sigma_1,\sigma_2)}(x)=\cost(\rho)
=\sum^{t-1}_{i=0}(\delta_i r_{k_l} + c_{j_l})
\]

We have that $\sum^{t-1}_{i=0}\delta_i$ is at most $1-x$, because there are no reset actions, $r_{k_l}\leq r_{\max}$, by definition, and $\sum^{t-1}_{i=0}c_{j_l} \leq v'(k,1)$, by construction of $\sigma_1$. 

Hence 
\begin{align*}v_k^{(\sigma_1,\sigma_2)}(x) &\leq (1-(1-\frac{\epsilon}{2 r_{\max}})) r_{\max}+ v'(k,1)\\
&=\epsilon/2+v'(k,1)\leq \epsilon/2+v'(k,x)
\end{align*}

That is a contradiction.

Otherwise, if $x< 1-\frac{\epsilon}{2 r_{\max}}$, there are two cases. Either the play defined by $(\sigma_1,\sigma_2)$ at some point waits until time $x'\geq 1-\frac{\epsilon}{2 r_{\max}}$ or not. If not, then the play cost at least $v'(k,x)+\epsilon/2$ because player 1 has at all times followed a strategy that guarantees at least that.

Otherwise, we can divide $\rho$ up in two. $\rho^1$ is the first part. The second part, $\rho^2$ begins in some state $k'$ and at time $x'$ such that $x'=1-\frac{\epsilon}{2 r_{\max}}$. Note that this might be in the middle of a wait period. Clearly $\cost(\rho)=\cost(\rho^1)+\cost(\rho^2)$. We must have that $\cost(\rho^1)+v'(k',x')\leq v'(k,x)+\epsilon/2$, because we followed a $\epsilon$ optimal strategy for player 1 in $G'$ in $\rho^1$. By the first part we also know that $\cost(\rho^2)\leq v'(k',x')+\epsilon/2$. 

Hence it is easy to see that $\cost(\rho)\leq v'(k,x)+\epsilon$. That is a contraction.

The construction of $\sigma_2$ can be done symmetrically.

\end{proof}

We are now ready for our reduction to SPTGs. 

\begin{lemma}\label{lem:remove intervals 3}
Solving any game $G$ in $G'\in PTG_{[0,1]}^{n,m}$ can be polynomially
reduced to solving an SPTG with  $n+1$ states and $m+1$ actions.
\end{lemma}

\begin{proof}
Player 2 will never use a action, $j$, to $\perp$ except at time 1 in a simple priced timed game, because player 2 might as well wait until time 1 before using $j$, which will not decrease the cost because rates are non-negative. Hence we can change all actions, $j$, where the existence interval is of the form $[1,1]$ to $[0,1]$ if $j\in A_k,k\in S_2$ (remember that if $j$ has an existence interval of the form $[1:1]$, then $d_j=\perp$, by definition of $PTG_{[0,1]}^{n,m}$), without changing the value functions. We will create a new state, $\max$, with the maximum rate in the game, belonging to player 2, which has a action to $\perp$ of cost 0 and existence interval $[0,1]$. We will now redirect all remaining actions (thus going from a state in $S_1$) which have existence interval $[1,1]$ to $\max$ (from $\perp$) and change the existence interval to $[0,1]$. We can see that player 1 will only use the actions to $\max$ at time 1, since it is cheaper to wait to time 1 and then move to $\max$.

Now all existence intervals have the form $[0,1]$. It is easy to see that we only need one action, $j$, for $j\in A_k$ and $d_j=\ell$ for any pair $k,\ell \in S$, because the controller of $k$ will, when playing optimally, only use the action with the best, for that player, cost. Hence we get an equvivalent simple priced timed game.

\end{proof}

\begin{lemma}\label{lem:ptg time}
Any game $G$ in $(n,m,r,d)$-PTG, can be solved in time $O((r+1)d(n\log n +\min (m,n^2)))$ using at most $(r+1)d$ calls to an oracle $R$ that solves SPTGs with $n+1$ states and at most $m+n+1$ actions.
\end{lemma}
\begin{proof}
The proof is a simple consequence of  Lemma \ref{lem:remove resets}, Lemma \ref{lem:remove intervals}, Lemma \ref{lem:remove intervals 2} and Lemma \ref{lem:remove intervals 3}.  
\end{proof}

Note that $d$ is bounded by $2m+1$ and $r$ is bounded by $n$.

\begin{theorem}\label{thm:final ptg time}
Any game $G$ in $(n,m,r,d)$-PTG, can be solved in time \[O((r+1)d(\min (m,n^2)+n\cdot\min\{12^n,\prod_{k\in S}(|A_k|+1)\})).\]
\end{theorem}
\begin{proof}
The proof is a consequence of Theorem \ref{thm:time} and  Lemma \ref{lem:ptg time}. Note that we only get $\prod_{k\in S}(|A_k|+1)$ and not $\prod_{k\in S}(|A_k|+2)$, because the additional actions we add to each state (using Definition \ref{def:sptgvxd} and \ref{def:Gxe}) both goes to $\perp$ and hence we only need one of them.
\end{proof}

\begin{theorem}\label{thm:bcfl-abm time}
The number of iterations used by the BCFL-ABM algorithm to solve any PTG $G$ is at most \[m\cdot n^{O(1)} \min\{12^n,\prod_{k\in S}(|A_k|+1)\}~.
\]
\end{theorem}
\begin{proof}
Note that Lemma \ref{thm:final ptg time} gives us a upper bound on the number of line segments of the value functions of $G$, because the number of line segments is a lower bound of the size of the output. 
By Bouyer {\em et al.}~\cite{BLMR06}, page 11, we know that the number of iterations needed for the BCFL-ABM algorithm is at most the number of line segments times $n$. 
\end{proof}

\begin{theorem}
\label{thm:reach}
Any priced timed game, $G$, in $(n,m,r,d)$-PTG, where all states have rate 1 and all actions have cost 0, can be solved in time $O((r+1)d(n\log n +\min (m,n^2)))$.
\end{theorem}

\begin{proof}
If we use Lemma \ref{lem:remove resets}, Lemma \ref{lem:remove intervals}, Lemma \ref{lem:remove intervals 2} and Lemma \ref{lem:remove intervals 3} on such a game, we get $(r+1)d$ SPTGs. If we look carefully at the lemmas we see that all states, in the SPTGs have rate $c$, for some $c>0$. $c$ depends on the interval length. Also, all actions that do not go to $\perp$ or $\max$ have cost 0. 

We now need to bound the number of event points. We will show that $\EventPoints(G')=1$ for $G'$ being any of the SPTGs generated. 

Look at the priced game $G^1$, as defined in section \ref{sec:solve sptgs}. Let $\sigma$ be some optimal strategy profile for $G^1$. We see that, if $r(k,\sigma)=0$ for some $k\in S$, we can not have passed through any states in $S_2$, since it is optimal to wait until time 1 for player 2. Since all actions of positive cost either goes to a state in $S_2$ or from a state in $S_2$, we must have that $v^\sigma_k=0$. 

For convenience we repeat the definition of the next event point and the function $f$ here.
The definition of the next event point, $\NextEventPoint(G^x)$, was:
\[
\max~\{0\}  \cup \{x' \in [0,x) \mid \exists k \in S, j \in
  A_k: f_{j,x}(x) \ne f_{\pi(k),x}(x) \wedge f_{j,x}(x') = f_{\pi(k),x}(x') \}.
\]

The definition of $f$ was
\[
f_{j,x}(x'') = c_j + a^x(\dest(j)) + b^x(\dest(j))(x - x'').
\]

Note that $b^x(\dest(j))$ corresponds to the rate of the next state we wait in if both players follow $\sigma$ and $f_{j,x}(x'')$ is the cost to reach $\perp$ if both players follow $\sigma$. Hence $f_{j,x}\geq 0$ and if $b^1=0$ then $f_{j,x}(x'')=0$ by the preceding, because $\sigma$ was optimal. Note that if $f_{j,1}(1) \ne f_{\pi(k),1}(1)$, then at least the larger expression of the two must have $b^1(\dest(j))=c$ and therefore we have that for all $x\in [0,1):f_{j,1}(x) \ne f_{\pi(k),1}(x)$, because either the $b^1(\dest(j))$'s are equal in the two expressions, in which case the difference between the two expressions do not change with $x$, or one is positive and the other is 0 for all $x\in [0,1]$.

Therefore we can apply Theorem \ref{thm:time} and get that
$\SolveSPTG$ solves any of $(r+1)d$ SPTGs in time $O(m+n\log n)$. We
therefore solve all in time $O((r+1)d(\min (m,n^2)+n\log n)$. The
reductions also required time $O((r+1)(n\log n +\min
(m,n^2))d)$. Hence, our time bound becomes $O((r+1)d(n\log n +\min (m,n^2)))$ .

\end{proof}

\begin{theorem}\label{thm:aut}
Any priced timed automata (i.e., all states are controlled by Player 1), $G$ in $(n,m,r,d)$-PTG, can be solved in time $O((r+1)d n^2(\min (m,n^2)+n\log n))$.
\end{theorem}

\begin{proof}
If we use Lemma \ref{lem:remove resets}, Lemma \ref{lem:remove intervals} and Lemma \ref{lem:remove intervals 2} on a priced timed automata, we get $(r+1)d$ priced timed games, without resets where all existence intervals are either $[0,1]$ or $[1,1]$ and all states belong to Player 1. The algorithm described in Figure \ref{fig:solve}, solves SPTGs by first solving them for time 1, as a priced game, and then solve them by induction backwards through time. We can still solve the game at time 1, and there is no differences in the induction, since no actions become available at time $x$ for $x<1$.
Hence, from Theorem \ref{thm:exp} and Theorem \ref{thm:time} we get
that we can solve such games in time $n^2(\min (m,n^2)+n\log n)$. The
reductions also required time $O((r+1)d(\min (m,n^2)+n\log n))$ and
therefore the total time complexity is  $O((r+1)d n^2(\min (m,n^2)+n\log
n))$.
\end{proof}

\section{Concluding remarks}

We have presented
an algorithm for solving one clock priced timed games with a complexity which is
close to linear in $L$, with $L = \EventPoints(G)$ being a lower bound
on the size
of the object to be produced as output. We think it is an attractive candidate for
implementation.

We have also given a new upper bound on $L$. While it is better than previous
bounds, we do not expect this bound to be optimal. It seems to be a
``folklore theorem'' that $L$ does not become very big for games
arising in practice. We would like to suggest the following
conjecture.
\begin{conjecture}
For all SPTGs $G$, $\EventPoints(G) \le p(n)$ for some polynomial $p$.
\end{conjecture}
Note that if this conjecture is established, it implies that our
algorithm as well as the BCFL-ABM algorithm runs in time polynomial in
the size of its input. 

\section{Acknowledgements}
We would like to thank Kim Guldstrand Larsen for many helpful discussions and comments.

\bibliographystyle{abbrv}
\bibliography{../priced}

\end{document}